\documentclass[11pt]{article}

\usepackage[a4paper, margin=1in]{geometry}
\usepackage{amsmath,amssymb,amsthm,amsfonts,cite,xcolor,xspace,graphicx,xifthen,bm}
\usepackage{thmtools,thm-restate}
\usepackage{appendix}
\usepackage{algorithm}
\usepackage[noend]{algpseudocode}
\usepackage[textsize=tiny]{todonotes}
\usepackage{enumitem}
\usepackage{bold-extra}
\usepackage[hang]{footmisc}
\setitemize{noitemsep,topsep=3pt,parsep=3pt,partopsep=3pt}

\definecolor{darkgreen}{rgb}{0,0.5,0}
\definecolor{darkblue}{rgb}{0,0,0.8}
\usepackage{hyperref}
\hypersetup{
    unicode=false,          
    colorlinks=true,        
    linkcolor=darkblue,          
    citecolor=darkgreen,        
    filecolor=magenta,      
    urlcolor=cyan,          
    hypertexnames=false
}

\usepackage[nameinlink,capitalize]{cleveref}

\newtheorem{theorem}{Theorem}[section]
\newtheorem{lemma}[theorem]{Lemma}

\theoremstyle{definition}
\newtheorem{definition}{Definition}[section]
\theoremstyle{remark}
\newtheorem*{remark}{Remark}

\newcommand{\mypara}[1]{
\smallskip

\noindent\textit{#1}\xspace}

\newcommand{\set}[1]{\left\{#1\right\}}
\newcommand{\bigO}{O}
\newcommand{\EV}{\ensuremath{\mathbb{E}}}
\newcommand{\calT}{\ensuremath{\mathfrak{T}}}
\newcommand{\calP}{\ensuremath{\mathcal{P}}}
\newcommand{\calC}{\ensuremath{\mathcal{C}}}
\newcommand{\NH}{\ensuremath{N}}
\newcommand{\UN}{\ensuremath{U}}

\newcommand{\Luby}{\textsc{DMis}\xspace}
\newcommand{\StableMIS}{\textsc{SMis}\xspace}
\newcommand{\DColor}{\textsc{DColor}\xspace}
\newcommand{\SColor}{\textsc{SColor}\xspace}

\newcommand{\Alg}{\ensuremath{\mathcal{A}}\xspace}
\newcommand{\StableAlg}{\textsc{SAlg}\xspace}
\newcommand{\DynamicAlg}{\textsc{DAlg}\xspace}
\newcommand{\CombinationAlg}{\textsc{Concat}\xspace}
\newcommand{\Trivial}{\textsc{Basic Coloring}\xspace}
\newcommand{\ch}[1]{\underline{#1}}

\newcommand{\algrule}[1][.2pt]{\par\vskip.5\baselineskip\hrule height #1\par\vskip.5\baselineskip}

\algnewcommand{\IfThenElse}[3]{
  \State \algorithmicif\ #1\ \algorithmicthen\ #2\ \algorithmicelse\ #3}
	
\newcommand{\algheading}[1]{
\algrule

\smallskip
\textbf{\boldmath #1}
}

\newcommand{\T}{T}
\newcommand{\Tcount}{R}
\newcommand{\poly}{\text{poly}}
\newcommand{\inparallel}{\textbf{parallel} }
\newcommand{\IS}{\ensuremath{\mathcal{M_{P}}}}
\newcommand{\DS}{\ensuremath{\mathcal{M_{C}}}}
\newcommand{\IC}{\ensuremath{\mathcal{C_{P}}}}
\newcommand{\DC}{\ensuremath{\mathcal{C_{C}}}}
\newcommand{\staticRadius}{\alpha}

\newcommand{\Hzero}{\ensuremath{H_r}}
\newcommand{\Htwo}{\ensuremath{H_{r+2}}}
\newcommand{\makeitemnice}[1]{\textbf{#1}}

\newcommand{\mis}{\emph{mis}\xspace}
\newcommand{\dominated}{\emph{dominated}\xspace}
\newcommand{\undecided}{\emph{undecided}\xspace}

\newcommand{\notification}{\ensuremath{mark}\xspace}

\algnewcommand\algorithmicswitch{\textbf{switch}}
\algnewcommand\algorithmiccase{\textbf{case}}

\algdef{SE}[SWITCH]{Switch}{EndSwitch}[1]{\algorithmicswitch\ #1\ \algorithmicdo}{\algorithmicend\ \algorithmicswitch}
\algdef{SE}[CASE]{Case}{EndCase}[1]{\algorithmiccase\ #1:}{\algorithmicend\ \algorithmiccase}
\algtext*{EndSwitch} 
\algtext*{EndCase} 

\algnewcommand{\NoAlignComment}[1]{\hspace{1cm}\textit{// #1}}

\newcommand{\hide}[1]{}

\begin{document}

\begin{flushleft}

\vspace*{0.8cm}
{\huge\bf Local Distributed Algorithms in \linebreak Highly Dynamic Networks}
\vspace{1.0cm}
\end{flushleft}

\newcommand{\auth}[3]{\textbf{#1}$\,\,\,\cdot\,\,\,$#2$\,\,\,\cdot\,\,\,$#3\par\medskip}

\auth{Philipp Bamberger}
{University of Freiburg}
{philipp.bamberger@cs.uni-freiburg.de}
\auth{Fabian Kuhn\footnote{Supported by ERC Grant No.\ 336495 (ACDC).}}
{University of Freiburg}
{kuhn@cs.uni-freiburg.de}
\auth{Yannic Maus\footnotemark[1]}
{University of Freiburg}
{yannic.maus@cs.uni-freiburg.de}

\vspace{1cm}

\begin{abstract}
We define a generalization of local distributed graph problems to (synchronous round-based) dynamic networks and present a framework for developing algorithms for these problems. The algorithms should satisfy non-trivial guarantees in every round. The guarantees should be stronger the more stable the graph has been during the last few rounds and coincide with the definition of the static graph problem if no topological change appeared recently. Moreover, if only a constant neighborhood around some part of the graph is stable during an interval, the algorithms should quickly converge to a solution for this part of the graph that remains unchanged throughout the interval.

We demonstrate our generic framework with two classic distributed graph problems, namely \emph{(degree+1)-vertex coloring} and \emph{maximal independent set (MIS)}. 
To illustrate the given guarantees consider the vertex coloring problem: Any conflict between two nodes caused by a newly inserted edge is resolved within $T=\bigO(\log n)$ rounds. During this conflict resolving both nodes always output colors that are not in conflict with their respective `old` neighbors. The largest color that a node is allowed to output is determined by the number of distinct neighbors that it has seen in the last $T$ rounds.
\end{abstract}
\setcounter{page}{0}
\thispagestyle{empty}
\newpage

\section{Introduction \& Related Work}

Many modern computer systems are built on top of large-scale networks such as the Internet, the world wide web, wireless ad hoc and sensor networks, or peer-to-peer networks. Often, the network topology of such systems is inherently dynamic: nodes can join or leave at any time and (e.g., in the context of overlay networks or mobile wireless networks) communication links might appear and disappear constantly. As a consequence, we aim to develop distributed algorithms that can cope with a potentially \emph{highly dynamic network topology} and to understand what can and what cannot be computed in a dynamic network. In particular, for \emph{local distributed graph problems} such as computing a graph coloring or a maximal independent set (MIS) of the network graph (see, e.g., \cite{alon86,ghaffari16,luby86abbrv,barenboimelkin_book}), we present a framework that allows to transform static problems and distributed algorithms into corresponding problems and algorithms for dynamic networks. 

Clearly, in an arbitrarily dynamic graph, it is not possible to always output a valid solution for the current network topology for any non-trivial graph problem. To overcome this problem most previous work on solving distributed graph problems in dynamic graphs is of the following flavor \cite{Censor-Hillel-MIS,onak18,whitbeck12}: After one or more topology changes, the algorithm has a \emph{recovery period} to fix its output and the network does not undergo any changes during this recovery period. However, if the network is highly dynamic, that is, further dynamic changes occur while recovering from a previous change, such an algorithm loses its guarantees and it might even fail to provide any guarantees at all. We therefore follow a different approach. 
We require that algorithms constantly adapt to a changing environment. They should always satisfy non-trivial guarantees, no matter how frequently the topology changes. The guarantees should become stronger if the network is less dynamic. In particular, if the network becomes static in a constant neighborhood around some part of the network, the solution of that part should also converge to a solution of the static graph problem after a short time and not change as long as the network remains locally static. Lastly, algorithms should work if the nodes wake up in an asynchronous fashion.

\mypara{Our Guarantees through the Lens of Coloring:}
The algorithms produced by our framework meet the aforementioned requirements and we apply it to two of the classic distributed graph problems, namely, the problem of computing a maximal independent set (MIS) and the problem of computing a vertex coloring of the network graph. We use this paragraph to  explain our guarantees by the example of the coloring problem; we however note that the general framework also applies to various additional graph problems. It seems for example particularly suitable to convert classic covering or packing optimization problems to the dynamic setting. Examples for such problems 
are minimum dominating set, minimum vertex cover, or maximum matching.\footnote{Our framework will actually require that the validity of a solution to a given problem can be checked locally. This helps to fix things locally. While the feasibility of the mentioned approximation problems can be checked locally, the guarantee on the size of the solution cannot be checked locally. However, in many cases, it is possible to consider a slightly extended problem for which the solution can be checked locally. For problems that can be phrased as linear programs, it is for example conceivable to consider a relaxed variant of the complementary slackness condition to locally verify the quality of a solution.} 
For the coloring problem, our algorithm guarantees that after two nodes are joined by an edge, they can only have the same color for a short time. Further, the total number of colors used is still essentially upper bounded by the maximum degree of the network as in the classic static version of the problem. In the context of dynamic networks, the degree of some node $v$ at a time $t$ is defined to be the number of distinct neighbors $v$ has had during the last few rounds.
Clearly, if all edges in some constant neighborhood are present in one round and non-present in the next round, the guarantees are weak and almost any output satisfies them. However, we believe that in applications usually only a small fraction of edges in some part of the graph changes such that our guarantees remain meaningful. For the coloring problem this means that the number of neighbors with the same color is always very small which is sufficient to resolve any conflict at a low cost with a simple randomized contention resolution strategy. In this context, we also want to emphasize that \emph{highly} dynamic networks do not refer to a huge amount of edges that change in every round but rather to the frequency of potential changes, i.e., changes can occur in every round and algorithms always have to provide guarantees---they cannot rely on a recovery time in which no changes occur.

\mypara{Related Work on Distributed Algorithms in Dynamic Networks:} By now, there is already a significant body of work that studies distributed computations in dynamic networks. However, to a large extent, the existing work deals with distributed solutions for mostly global network problems such as broadcasting information to all nodes of a dynamic network \cite{ahmadi15,augustine16,baumann09,casteigts15,clementi09,dutta13,haeupler11,dynnetworks,odell05}, computing a global function on inputs that are distributed among the nodes of a dynamic network \cite{jahja17,dynnetworks,michail13,michail14,santoro15,yu16}, performing a random walk on the nodes of a dynamic network \cite{avin08,sarma15,denysyuk14}, solving agreement problems in dynamic network \cite{augustine15,kuhn11:coordconsensus,olfati-saber04,ren04,welch09leader}, or synchronizing clocks in a dynamic network \cite{fuegger15,KuhnLocherOshman11,kuhn10_clocksynch}. 

Even though the concept of locality cannot immediately be transferred to dynamic graphs\footnote{The concept of locality can be redefined for dynamic networks using time-expanded graphs, see, e.g., \cite{KO11}.}, we believe that local distributed algorithms in static networks \cite{peleg00} are particularly suited for dynamic networks:
If a distributed algorithm has time complexity $T$ in a static network $G$, the output of each node $v$ only depends on the initial state of the $T$-neighborhood of $v$ in $G$. Therefore if the topology of $G$ only changes locally, the algorithm can be used to repair an existing solution in time $T$ by only changing the output of nodes in a $T$-neighborhood around the local topological changes.\footnote{The statement holds for deterministic algorithms and a weaker version holds for randomized algorithms.}  In our opinion, this fact is one of the key motivations for the everlasting search for distributed algorithms that are as local as possible. In \cite{awerbuch88,awerbuch92,lenzen09}, this connection between local algorithms and dynamic networks is made explicit. In \cite{awerbuch88,awerbuch92}, it is shown that a synchronous $T$-round algorithm can be run in an asynchronous dynamic network such that whenever the $T$-hop neighborhood of some part of the dynamic graph becomes stable, the algorithm also eventually converges to a stable solution in this part of the graph. We note that if the graph never becomes stable in some part, the results of \cite{awerbuch88,awerbuch92} do not guarantee anything.
In \cite{lenzen09}, it is shown that local distributed algorithms can be turned into fast converging self-stabilizing algorithms.\footnote{A distributed algorithm is called self-stabilizing if it is guaranteed to converge to a stable and valid solution (in a static network) even if the algorithm starts in an arbitrary initial state \cite{dijkstra74,dolev00}.} The problem of locally repairing a single dynamic change in the network has been studied in \cite{Censor-Hillel-MIS} for the problem of computing an MIS. They show that a simple randomized distributed greedy algorithm guarantees that when a single topological change occurs (i.e., if a single node or edge is inserted or deleted), on average, the MIS can be repaired in constant time and in fact even such that only a constant number of nodes need to change their state. Just recently  this result was even strengthened by the development of a \emph{deterministic} distributed algorithm with constant amortized round and adjustment complexity \cite{onak18,gupta18,du18}.
While the above results certainly encourage the use of local algorithms in dynamic networks, they do not show that such algorithms can be used to always produce a meaningful output in a dynamic network with constant topological changes.

\subsection{Contribution \& Techniques in a Nutshell}
\label{ssec:contribution}

The contribution of this paper is threefold. We define a general method to turn a large class of static graph problems into graph problems that are defined on arbitrarily dynamic graphs. The valid outputs at any point in time are defined by the dynamic graph topology of the last $T$ time units, where $T$ is a parameter that ideally is at most polylogarithmic in the number of nodes. We further provide a framework that allows to develop distributed algorithms for these problems. 
Then, we modify known algorithms for static graphs for two sample problems (MIS and coloring) to demonstrate that the framework can be used (almost in a black-box manner) with such existing algorithms. This strengthens the aforementioned statement on the usefulness of local algorithms for static graphs in the dynamic setting: Now, with our framework such algorithms can be used to repair solutions while always providing non-trivial guarantees, even during the repair process and no matter how frequently changes occur.
In the following, we provide an informal description of our model and framework, for formal definitions, we refer to \Cref{sec:model,sec:framework}. 

We model a dynamic network as a synchronous system over a set $V$ of $n$ potential nodes. Time is divided into rounds and in each round $r=0,1,2,\dots$, there is a communication graph $G_r=(V_r,E_r)$. We will later assume that nodes can wake up gradually, however for the purpose of this summary, we assume that all nodes wake up initially and we thus have $V_r=V$ for all $r\geq 1$. We consider graph problems that can be decomposed into two parts that are given by a \emph{packing} and a \emph{covering} graph property. Essentially, a packing property is a graph property that remains true when removing edges and a covering property is a graph property that remains true when adding edges. In addition, we assume that the validity of a solution can be checked locally, i.e., by evaluating it in the constant neighborhood of every node \cite{LD11,NaorStockmeyer93}. 
For example, the problem of finding an MIS on a graph $G$ can be decomposed into the problem of finding a subset $S$ of the nodes such that no two neighbors are in $S$ (packing property) and $S$ is a dominating set of $G$ (covering property). For the (degree+1)-coloring problem, the requirement that the vertex coloring is proper is a packing property and the requirement that the color of a node $v$ is from $\set{1,\dots,\deg(v)+1}$ is a covering property. For a given graph problem and an integer parameter $\T\geq 1$, we say that a given solution is a \emph{$\T$-dynamic solution} at time $r$ if a) the solution satisfies the packing property for the \emph{intersection graph} $G_r^{\T\cap}=G_{r-\T+1}\cap G_{r-\T+1}\cap \ldots \cap G_{r}$ (i.e., the graph that contains all edges that have been present throughout the last $\T$ rounds), and b) the solution satisfies the covering property for the \emph{union graph}\footnote{The idea to describe the feasibility of covering solutions with the help of union graphs already appeared in the introduction of \cite{casteigts11} as the \emph{over-time} variant of a dynamic graph problem. However, the paper suggests to take the union of all graphs that have appeared until the current time slot. Our approach is much more local in time as we move a sliding window on the sequence of graphs and the feasibility of an output only depends on the graphs that are in the current sliding window.} $G_r^{\T\cup}=G_{r-\T+1}\cup G_{r-\T+1}\cup \ldots \cup G_{r}$ (i.e., the graph that contains all edges that have been present at least once in the last $\T$ rounds).  

When designing a distributed algorithm for a given dynamic graph problem, we require that for some $\T\geq 1$, the algorithm outputs a $\T$-dynamic solution after each round $r$. Assume that we can construct an algorithm $\mathcal{A}$ such that if all nodes start $\mathcal{A}$ in round $1$, after round $\T$, $\mathcal{A}$ outputs a $\T$-dynamic solution w.r.t.\ to the first $\T$ graphs (i.e., a solution that satisfies the packing property for $G_{\T}^{\T\cap}$ and the covering property for $G_{\T}^{\T\cup}$). Given such an algorithm $\mathcal{A}$, we can in principle design an algorithm that always outputs a $\T$-dynamic solution by just starting a new instance of $\mathcal{A}$ in every round and outputting the solution of an instance started in round $r+1$ after round $r+\T$. However, clearly such a solution would not be satisfactory because especially if $\mathcal{A}$ is randomized, the output might change completely from round to round even if the graph is only mildly dynamic or even static. Thus, we also require that the output does locally not change if the graph is static in some local neighborhood.  If the graph has been static during rounds $r-\T+1,\dots,r$, a $\T$-dynamic solution at time $r$ is a non changing solution of the static graph problem for the graph $G_r$ in round $r$. We believe that the concept of a $\T$-dynamic solution that is locally static if the graph is locally static provides a natural generalization of a static graph problem to the dynamic context.

In order to simplify the process of finding new algorithms we develop a framework that separates the two tasks of (1) always outputting a $\T$-dynamic solution and (2) providing a locally stable output if the network is locally static. Therefore, we define two \emph{abstract} types of algorithms. For two positive integers $\T$ and $\alpha$, we say that an algorithm $\StableAlg$ is a \emph{$(\T,\alpha)$-network-static algorithm} for a given dynamic graph problem if it satisfies the following properties. At the end of each round $r\geq 1$, the algorithm outputs a valid \emph{partial} solution for the graph $G_r$.\footnote{In a partial solution, nodes are allowed to output $\bot$. For each node $v$ that outputs a value $\neq \bot$, it must hold that there exists an extension of the partial solution such that the packing property for $v$ is satisfied and the covering property for $v$ is satisfied for all extensions of the partial solution. For a formal definition, we again refer to \Cref{sec:model,sec:framework}.} In addition, if the $\alpha$-neighborhood of some node $v$ remains static in some interval $[r,r_2]$, $v$ must output a fixed value $\neq\bot$ throughout the interval $[r+T,r_2]$. Further, for a positive integer $\T$, we say that an algorithm $\DynamicAlg$ is a \emph{$\T$-dynamic algorithm} for a given dynamic graph problem if it satisfies the following property. Let $r\geq 1$ be some round and assume that we are given a valid partial solution for $G_r$. If $\DynamicAlg$ is started in round $r+1$, at the end of round $r+T-1$, it outputs a $\T$-dynamic solution that extends the given partial solution for $G_r$. The following theorem shows that a $\T_1$-dynamic algorithm and a $(\T_2,\alpha)$-network-static algorithm can be combined to obtain a distributed algorithm that always outputs a $\T_1$-dynamic solution while (essentially) inheriting the properties of $\StableAlg$ if the graph is locally static for sufficiently long. 

\begin{theorem}
	\label[theorem]{thm:mainPackingCovering}
	Let $\T_1$ and $\T_2$ be positive integers, $\calP$ a packing, 
	and $\calC$ a covering problem.
	Given a $\T_1$-dynamic algorithm and a $(\T_2,\alpha)$-network-static algorithm for $(\calP,\calC)$, one can combine both algorithms to an algorithm such that:
	\begin{enumerate}
		\item (dynamic solution) Its output in round $r$ is a $\T_1$-dynamic solution for $(\calP,\calC)$~. 
		\item (locally static) If the graph is static in the $\alpha$-neighborhood of a node $v\in V_r$ in all rounds in an interval $[r,r_2]$ then the output of $v$ does not change for all rounds in $[r+\T_1+\T_2,r_2]$.
	\end{enumerate}
\end{theorem}

The significance of a $T$-dynamic solution gets stronger the smaller $T$ is chosen (for any $T'>T$, a $T$-dynamic solution is also a $T'$-dynamic solution, but not vice versa). On the other hand, to obtain an algorithm that outputs a $T$-dynamic solution in some round $r$ for any graph sequence, $T$ must be at least as large as any lower bound on the time to solve both the packing and covering problem on static graphs. To see this, assume $T$ is smaller than such a lower bound and we have an algorithm that outputs a $T$-dynamic solution in round $r\geq T$ for any given graph sequence. Then, for any graph $G$, consider the graph sequence which consists of the empty graph in all rounds up to $r-T$ and of $G$ in all rounds afterwards. Then a $T$-dynamic solution in round $r$ is a solution for both the packing and covering problem in $G$, which means that the algorithm computed  a solution in $T$ rounds (as it has no knowledge on the edges of  $G$ before round $r-T$). Conditioned on the currently known runtimes  (expressed as a function of $n$) being optimal \cite{ghaffari16}, our window size for MIS (cf. \Cref{cor:mainMIS}) is optimal.

\subsection{Two Sample Problems: MIS \& Vertex-Coloring}
We show how to apply the above framework to two of the classic local symmetry breaking problems: computing a vertex coloring and computing an MIS of the network graph. In both cases, we adapt existing randomized algorithms to obtain the algorithms that are required for the framework.
For vertex coloring, we use a variant of the most basic randomized coloring algorithm. In each round, each uncolored node $v$ selects a uniformly random color from $\set{1,\dots,\deg(v)+1}\setminus S$, where $S$ is the set of colors that are already taken by the colored neighbors of $v$. Node $v$ keeps a color if no neighbor chooses the color in the same round.\footnote{It is commonly known that this simple randomized algorithm terminates in $O(\log n)$ rounds in static graphs. The algorithm is for example used and analyzed in \cite{barenboim12,johansson99}.}

\begin{restatable}{corollary}{mainColoring}
\label{cor:mainColoring}
There is a $\T=O(\log n)$ and an algorithm that, w.h.p., outputs a $\T$-dynamic solution for (degree+1)-coloring in every round and  the output of any node $v$ is static in all rounds in the interval $[r+2\T,r_2]$  if the $2$-neighborhood of $v$ is static in all rounds $l\in[r,r_2]$.
\end{restatable}

We say that a statement holds with high probability (w.h.p.) if it holds with probability $1-1/n^c$ for a constant $c>1$ that can be chosen arbitrarily. We assume that all executions are of length at most polynomial in $n$. All our probabilistic results could be extended to arbitrarily long executions if we allow the output to be invalid in a polynomial small fraction of the rounds.

For the MIS problem we adapt the algorithm by Ghaffari \cite{ghaffari16} to obtain a $(O(\log n), O(1))$-network-static algorithm \StableAlg and we adapt Luby's well-known algorithm \cite{alon86,luby86abbrv} to obtain a $O(\log n)$-dynamic algorithm \DynamicAlg.

\begin{restatable}{corollary}{mainMIS}
\label{cor:mainMIS}
There is a $\T=O(\log n)$ and an algorithm that, w.h.p., outputs a $\T$-dynamic solution for MIS in every round and  the output of any node $v$ is static in all rounds in the interval $[r+2\T,r_2]$  if the $2$-neighborhood of $v$ is static in all rounds in the interval $[r,r_2]$.
\end{restatable}

We see the simple adaptation\footnote{Of course, some of the existing proofs need additional care and some algorithms, e.g., the MIS algorithm by Ghaffari, need some (crucial) modifications to assure termination in the dynamic setting.}---compared to a huge and heavy machinery--- of existing static algorithms to the dynamic case as a strength of the framework in terms of practicability. 

\mypara{Relevance of MIS and Vertex Coloring in Dynamic Networks:}  We believe that in particular MIS and vertex coloring are natural problems to study in a dynamic network context. They are the prototypical problems to study the challenge of local symmetry breaking in distributed network algorithms, they are among the most thoroughly studied problems of the area, and they are important building blocks in various other distributed algorithms \cite{awerbuch1989network,panconesi95,KS17}. 
Apart from this, some of the standard applications of MIS and coloring are in the context of networking scenarios where networks are likely to exhibit some dynamics. For example, an MIS is often used to obtain some local centers or some basic clustering of the network, specifically also in the context of wireless networks \cite{MosciMIS05}. In fact, the problem of selecting a subset of management/monitoring nodes within dynamic networks has also been studied in much more applied contexts, e.g., \cite{clegg_manage_elect_2013} develops heuristic algorithms for the problem and evaluates their performance on real world dynamic graphs. The standard application of vertex coloring is to assign frequencies or time slots to the nodes of a network in order to coordinate the access to a shared channel.  This setting is also helpful to interpret our guarantees that, combined with a simple randomized contention resolution strategy,  can be used for such an assignment.
\subsection{Alternative Approaches to Study Highly Dynamic Networks}
\label{ssec:furtherRelated}
Besides the intensively studied synchronous round based dynamic graphs \cite{dynnetworks,Censor-Hillel-MIS,onak18} so called more general (discrete or continuous) \emph{time varying graphs} are studied with (asynchronous) message passing \cite{tvg12,dubois15}. The downside of the \emph{recovery time approach} for highly dynamic networks was identified in \cite{dubois15} and to still produce meaningful output authors either (1) restrict the allowed topological changes \cite{dynnetworks}, or (2) change the objective of algorithms \cite{dubois15}. 
The taste of the latter approach can be illustrated by \cite{dubois15} where algorithms compute a single  set $M$ that is a dominating set of the so called \emph{footprint graph $G^{\omega}$}. Here, the graph $G^{\omega}$ only consists of those edges that appear infinitely often in the dynamic sequence of graphs. The runtime of an algorithm in this model is the time until the output converges to a stable solution---this is clearly incomparable to the runtimes of our algorithms. As the graph $G^{\omega}$, for which the algorithm computes a solution, depends on the whole infinite sequence of graphs there are no guarantees on the output if we only look at the behavior of the algorithm in some small time window. In contrast, our notion of a $\T$-dynamic solution gives these guarantees: one can see our approach as a sliding window that moves throughout time and the feasibility of our output always depends on the graphs in the current sliding window.

\subsection{Outline}

In \Cref{sec:model} we formally define our dynamic graph model and formalize the notion of dynamic distributed graph problems. In \Cref{sec:framework} we formally define packing and covering graph problems, $T$-dynamic and $(T,\alpha)$-network-static algorithms and prove \Cref{thm:mainPackingCovering}. In \Cref{sec:coloring} and \Cref{sec:algorithm} we apply our methods to the (degree+1)-coloring and the MIS problem. In \Cref{sec:discussion} we discuss our results and point out further research.

\section{Dynamic Graph Model}
\label{sec:model}
A dynamic graph  is a sequence of graphs $G_0=(V_0,E_0),G_1=(V_1,E_1), G_2=(V_2,E_2), \ldots$ that is provided by a worst case adversary in a synchronous round-based model. We require that the sequence of nodes $\emptyset=V_0\subseteq V_1\subseteq  \ldots$ is increasing. This allows the addition of nodes to the network and a node $v$ leaving the network can be modeled by removing all edges adjacent to $v$ but keeping the node in the network as an inactive isolated node. Throughout this work $n$ is an upper bound on the number of nodes in $V_i$ for each $i$ and $n$ is known by all nodes of the network. Round $r$ consists of the following steps:
\label{sec:dynModel}
\begin{enumerate}
\item The adversary changes the graph, i.e., it provides graph $G_r=(V_r,E_r)$,
\item Nodes send/receive messages through the edges $E_r$ and perform local computations,
\item Each node returns its output. 
\end{enumerate}
The algorithm can use fresh randomness in every round. The communication is by \emph{local broadcast} and a node does not have to know its neighbors at the beginning of a round; in particular a node does not know its degree in $G_r$ at the beginning of round $r$. We do not limit the message size but all presented algorithms can be adapted to work with $\poly\log n$ bits per message.
Whenever we say that a property holds \emph{in round $r$} we mean that the property holds at the end of round $r$, that is, before the adversary has changed the input graph to $G_{r+1}$ and after the nodes have performed the computations of round $r$. 

\begin{definition} \label[definition]{Gcap}
For any integer $\T\geq 0$ and round $r$, define $r_0=\max\{0,r-\T+1\}$ and
\begin{align*}
V^{\T\cap}_r := \bigcap_{r'=r_0}^{r} V_{r'}\quad& \text{and} & \quad
E^{\T\cap}_r := \bigcap_{r'=r_0}^{r} E_{r'}\quad& \text{and} & \quad
E^{\T\cup}_r := \bigcup_{r'=r_0}^{r} E_{r'}~.
\end{align*}
We call $G^{\T\cap}_r := \left(V^{\T\cap}_r,E^{\T\cap}_r\right)$ the \emph{($\T$-)intersection graph (in round $r$)} and  $G^{\T\cup}_r := \left(V^{\T\cap}_r,E^{\T\cup}_r\right)$ the \emph{($\T$-)union graph (in round $r$)}.
\end{definition}
We use the aforementioned graphs to transfer distributed graph problems for the static setting to the dynamic setting where the feasibility of a solution depends on the union (intersection) graph of the last few rounds (cf. the definition of a $T$-dynamic solution in \Cref{ssec:contribution}). We want to mention that the idea to transfer a covering graph problem to the dynamic setting  by defining a solution with respect to the union of the whole graph sequence appeared in the introduction of \cite{casteigts11} (but was not further used in the paper). The main difference of our approach is that it is much more local in time as we move a sliding window on the sequence of graphs and the feasibility of an output only depends on the graphs that are in the current sliding window---typically we imagine the sliding window to be small, that is, we only union the graphs of the last few rounds and obtain guarantees that only depend on the topological changes in the last few rounds. Note that for a covering graph problem the feasibility of an output for a small time window always implies the feasibility for a larger window and, in particular, the feasibility  with regard to the union of the whole sequence (cf. \Cref{def:packingcovering}).

\emph{Asynchronous wake up} can be modeled via $V_r$ being the nodes that have woken up until round $r$. Then, in round $r$,  $V^{\T\cap}_r$ contains the nodes that have been awake for at least $T$ rounds and $V_0=\emptyset$ means that all nodes are asleep at the beginning. When a node wakes up it does not know the current round number; round numbers are only for the sake of analysis.  
Note that $G^{\T\cap}_r\subseteq G_r$, so any edge in the intersection graph can be used for communication purposes in round $r$. However, there is no guarantee that edges in $G^{\T\cup}_r$ can be used for communication in round $r$.

\begin{definition}[Distributed Graph Problem]
\label[definition]{definition:graphProblem}
A \emph{distributed graph problem} $\calT$ is given by a set of tuples of the form $(G,\vec{y})$, where $G=(V,E)$ is a simple, undirected graph and $\vec{y}$ is a $|V|$-dimensional vector with entries $y_v$ for each node $v\in V$. The \emph{output vector} $\vec{y}$ is called a \emph{solution} for $\calT$ if $(G,\vec{y})\in \calT$. Furthermore, $y_v$ is the \emph{output} of $v$; if a node has not produced any output yet we set $y_v=\bot$. A vector $\vec{z}$ is called an \emph{extension} of $\vec{y}$ if  $z_v=y_v$ whenever $y_v\neq \bot$.
In a solution we require that all nodes produce some output. A vector $\phi$ with an entry for each node of $G$ is also called an \emph{input}.
\end{definition}

In this paper we consider distributed graph problems for which the feasibility of a solution can be verified by checking the solution for each $O(1)$-radius neighborhood (cf. the problem class $\text{LD}(O(1))$ in \cite{LD11}); maximal independent set and coloring can be checked with radius one. In the style of locally checkable labeling problems (LCL problems) \cite{NaorStockmeyer93} we say that the \emph{LCL condition} is satisfied for a node if the feasibility check of its $O(1)$-neighborhood is positive.
We model the maximal independent set (MIS) as all pairs $(G,\vec{y})$ such that $M=\{v \in V \mid y_v=1\}$ is an MIS of $G$ and $y_v=0$ for all $v\notin M$. The problem of properly $c$-coloring consists of all pairs $(G,\vec{y})$ with $y_v\in [c]$ for all $v\in V(G)$ and $y_v\neq y_u$ for all $\{u,v\}\in E(G)$.

A \emph{dynamic distributed graph problem} is given by a set of sequences $(G_1,y_1),(G_2,y_2),\dots$ where each $G_r$ is a simple, undirected graph and $y_r$ is a $|V_r|$-dimensional vector. The vector $y_r$ is interpreted as a feasible output or a solution in round $r$.

A \emph{$\rho$-oblivious adversary} does not know the random bits of the last $\rho$ rounds, e.g., a $2$-oblivious adversary does not know the random bits of round $r$ and $r-1$  when determining graph $G_r$. 
An \emph{adaptive offline adversary} knows all random bits of the algorithm in advance. Our algorithms rely on different types of adversaries and we mention the respective type with the respective algorithm.
For an algorithm $\Alg$ let $\Alg_r^{r'}(\phi)$ denote the output of the algorithm if it starts its computation in round $r$ with input $\phi$ and runs until round $r'$ (inclusively), that is, it executes the rounds $r,r+1,r+2,\ldots r'$. For a node $v\in V(G)$, $\NH_G(v)$ denotes the set of its neighbors in the graph $G$. For positive integers $\alpha,k$ and a node $v$, let $\NH_{\alpha}(v)$ denote the $\alpha$-neighborhood of $v$ and $[k]:=1,\dots,k$. For a round $r$, and a positive integer $T$, we denote by $d_r(v)$ ($d^{\cap\T}_r(v)$, $d^{\cup\T}_r(v)$ resp.) the degree of $v$ in $G_r$ ($G^{\cap\T}_r$, $G^{\cup\T}_r$ resp.).

\medskip

We repeatedly use the inequalities $(1-x)\leq e^{-x}$ for all $x$ and 
\begin{align}
\label{eqn:x4}
1-x\geq 4^{-x} \text{ for } x\leq\frac{1}{2}~.
\end{align}

\section{A Framework for Highly Dynamic Network Algorithms}
\label{sec:framework}

\label{sec:packingCovering}
The class of distributed graph problems that we transfer to the dynamic setting consists of problems that can be decomposed into a packing and a covering component.
\begin{definition}[Packing, Covering Problem]\label[definition]{def:packingcovering}
We call a distributed graph problem $\calT$ 
\begin{itemize}
\item \emph{packing} if any solution for a graph $G$ is a solution for any graph $G_1=(V,E'\subseteq E(G))$,
\item \emph{covering} if any solution for a graph $G$ is a solution for any graph $G_1=(V,E'\supseteq E(G))$.
\end{itemize}
\end{definition}

In a packing distributed graph problem (e.g, the independent set problem), edges can be seen as constraints on how much can be \emph{packed} (into the independent set) and removing constraints preserves the feasibility of a solution. In a covering distributed graph problem (e.g., the dominating set problem), edges help to \emph{cover} (nodes) and thus adding edges preserves the feasibility of a solution. These properties coincide with those of classical packing and covering problems, which motivates the terminology. As a further example, properly coloring without restriction on the number of colors is a packing problem. (Improperly) coloring a given graph where adjacent nodes are allowed to have the same colors and where $v$'s color is in the range $\set{1,\dots,\deg(v)+1}$ is a covering problem.

Very often packing and covering problems have trivial solutions, e.g., the empty set is an independent set or all nodes form a dominating set. In the setting of LCL problems usually only their intersection is an object of interest, e.g., the intersection of the independent set problem and the dominating set problem defines the MIS problem. The intersection of the introduced packing and covering coloring variants leads to the standard (degree+1) coloring problem.
Our goal is to devise algorithms for highly dynamic networks that, in every round, guarantee properties which are closely related to the original problem and behave well in static graphs. In particular we desire the following guarantees: (1) For a suitably chosen $\T$ and any round $r$, the output should be a solution for the packing problem in $G_r^{\cap\T}$ and for the covering problem in $G_r^{\cup\T}$; (2) the output should locally not change if the dynamic graph is locally static. 
We present a general framework to combine algorithms that separately take care of the requirements (1) and (2). 
The following natural properties describe the algorithms satisfying (1) and (2).
\begin{definition} \label[definition]{def:partialsolution}
Let $\calT$ be a distributed graph problem. We call a vector $\phi$ 
\begin{itemize}
\item \emph{partial packing} for $\calT$ if \underline{there is} an extension $\bar{\phi}$ of $\phi$ with $\bar{\phi}_u\neq \bot$ for all $u\in V$, such that for all nodes $v$ with $\phi_v\neq\bot$ the LCL condition of $\calT$ is satisfied in $\bar{\phi}$.
\item \emph{partial covering} for $\calT$ if \underline{for all} extensions $\bar{\phi}$ of $\phi$ with $\bar{\phi}_u\neq \bot$ for all $u\in V$ and for all nodes with $\phi_v\neq\bot$ the LCL condition of $\calT$ is satisfied in $\bar{\phi}$.
\end{itemize}
Let $\calP$ be a packing problem and $\calC$ a covering problem. We call an output vector $\phi$ a \emph{partial solution} for $(\calP,\calC)$ if $\phi$ is partial packing for $\calP$ and partial covering for $\calC$. 
\end{definition}

\begin{definition}[dynamic, network-static] \label[definition]{def:dynamic/static}
Let $\calP$ be a packing problem, $\calC$ a covering problem, $T$ and $\alpha$ positive integers and $G_0,G_1,\ldots$ a dynamic graph.
\begin{itemize}
\item An algorithm $\Alg$ is called \emph{$\T$-dynamic} for $(\calP, \calC)$ if it satisfies the following:
\begin{enumerate}
\item[\makeitemnice{A.1}] \textit{(input-extending)} For any $j'\geq j$ and any vector $\phi$, $\Alg_j^{j'}(\phi)$ is an extension of $\phi$.
\item[\makeitemnice{A.2}] \textit{(finalizing)} For $j\geq \T-1$ and any partial solution $\phi$ for $(\calP$,$\calC$) in $G_{j-\T+1}$, the output  
$\Alg_{j-\T+2}^{j}(\phi)$ is a solution for $\calP$ in $G_j^{\cap\T}$ and a solution for $\calC$ in $G_j^{\cup\T}$.
\end{enumerate}

\item An algorithm $\Alg$ is called \emph{$(\T,\alpha)$-network-static} for $(\calP, \calC)$ if it satisfies for any input $\phi$:
\begin{enumerate}
\item[\makeitemnice{B.1}] \textit{(partial solution)} Its output in round $j$ is a partial solution for $(\calP,\calC)$ in $G_j$.
\item[\makeitemnice{B.2}] \textit{(locally static)} For each $v\in V_r$ and each interval $[r,r_2]$ with $G_l\big[\NH_{\staticRadius}(v)\big]=G_{l'}\big[\NH_{\staticRadius}(v)\big]$ for all $l,l'\in [r,r_2]$, the output of $v$ is $\neq\bot$ and does not change for all $l\in[r+\T,r_2]$.
\end{enumerate}
\end{itemize}
\end{definition}
A.1 requires that a dynamic algorithm never deletes anything from a partial solution of a problem. A.2 says that any solution which is a partial solution of both problems is completed within $\T$ rounds.
B.1 ensures that the algorithm always computes partial solutions for the current graph and B.2 ensures that the algorithm behaves well if it is locally static.

Now, we combine a $\T_2$-network-static algorithm \StableAlg with a $\T_1$-dynamic algorithm \DynamicAlg. \StableAlg is started in round zero and serves as a base algorithm that first computes a partial solution and forwards it to \DynamicAlg. Then \DynamicAlg extends it to a full solution. If the graph is locally static, \StableAlg provides a locally unchanged output that is not changed by \DynamicAlg.

\begin{algorithm}[H]
\caption{Round $r$ of \CombinationAlg}
\small
\begin{tabular}{@{}lll}
\textbf{Input:} & $\bot$ & (no node has an output)\\
\textbf{Output:} & $\phi_r$ & \\
\textbf{Vars.: } 	&	$\phi_j$ &  Output of  \StableAlg in round $j$ (partial solution for $(\calP,\calC)$ in $G_j$)\\
											& $\mathcal{J}$ & One of the \DynamicAlg-instances
\end{tabular}\\

\medskip
  
\textbf{Start:} Initiate an $\StableAlg$-instance $\StableAlg(\bot)$; $\phi_{-1}=\bot$ \NoAlignComment{No communication round needed}
\algheading{Round $r$ of \CombinationAlg}
\begin{algorithmic}[1]
\State Start a new \DynamicAlg-instance $\DynamicAlg(\phi_{r-1})$
\If{there are $\T_1-1$ \DynamicAlg instances}
\State Discard the oldest \DynamicAlg-instance
\EndIf
\For{all \DynamicAlg-instances $\mathcal{J}$} in \inparallel
\State Execute one round of $\mathcal{J}$
\EndFor
\State In \textbf{parallel} to the above, execute one further round of $\StableAlg$; denote the output with $\phi_r$. 
\State\textbf{Output} the output of the oldest \DynamicAlg-instance
\end{algorithmic}
\label{combined}
\end{algorithm}

\begin{proof}[Proof of \Cref{thm:mainPackingCovering}]~
\begin{enumerate}
\item If $r<\T_1-1$, the graphs $G_r^{\cap\T_1}$ and $G_r^{\cup\T_1}$ are both empty as no node has been awake for $T_1$ rounds.
For $r\geq \T_1-1$ let $\psi:=\StableAlg_0^{r-\T_1+1}$. At the beginning of round $r-\T_1+2$, $\CombinationAlg$ starts a new instance of $\DynamicAlg$ on $\psi$. This instance becomes the output of $\CombinationAlg$ exactly after the run of $\T_1-1$ rounds, i.e., \CombinationAlg outputs $\DynamicAlg_{r-\T_1+2}^r(\psi)$ after round $r$. By property $B.1$ we know that $\psi$ is a partial solution for $(\calP,\calC)$ in $G_{r-\T_1+1}$ and thus by $A.2$ $\CombinationAlg_0^r$ is a solution for $\calP$ in $G_r^{\cap\T_1}$ and a solution for $\calC$ in $G_r^{\cup\T_1}$.

\item Due to $B.2$, we have $(\psi)_v:=(\StableAlg_0^{r+\T_2})_v=(\StableAlg_0^l)_v\neq \bot$ for all $l\in [r+\T_2,r_2]$. This is the input of all $\DynamicAlg$ instances starting between $r+\T_2+1$ and $r_2+1$. As $\DynamicAlg$ is input-extending (A.1) $(\psi)_v$ is also the output of $\CombinationAlg$ for $v$ in any round $l\in[r+\T_1+\T_2,r_2]$. \qedhere
\end{enumerate}
\end{proof}
The following remark holds along similar lines as the proof of \Cref{thm:mainPackingCovering}, 1.
\begin{remark}
\Cref{thm:mainPackingCovering} also holds if $V_0\neq\emptyset$ and the algorithm is started with a solution in $G_0$ for $\calP$ and $\calC$ as input.
\end{remark}

\begin{remark}
In principle, using the same technique, one could also combine more than two algorithms. One could for example imagine to also have a dynamic network algorithm that has stronger guarantees, but only works in dynamic networks with much more limited dynamic changes. In combination with the static and the dynamic algorithms considered in the paper, this can lead to an algorithm that a) converges to a locally stable solution if the graph is locally static, b) satisfies the stronger dynamic guarantees if the topological changes are only of the required limited form, and c) satisfies the dynamic guarantees of the present paper for arbitrary dynamic topologies.
\end{remark}

We complete this section by making our statement that we turn a large class of static graph problems into graph problems defined on dynamic graphs formal. For a static graph problem which can be decomposed into a packing problem $\calP$ and a covering problem $\calC$ and a parameter $\T$, the corresponding dynamic graph problem consists of all sequences $(G_1,y_1),(G_2,y_2),\dots$ such that each $y_r$ is a $\T$-dynamic solution for $(\calP,\calC)$, i.e., $(G^{\cap T}_r,y_r)\in\calP$ and $(G^{\cup T}_r,y_r)\in\calC$. Given $(\calP,\calC)$, our framework allows to build algorithms for the aforementioned corresponding dynamic graph problem with the additional property of giving a locally static solution if the graph is locally static.

\section{Coloring in Highly Dynamic Networks}
\label{sec:coloring}
In this chapter we consider the coloring problem. Let $\IC$ be the problem of properly  coloring  the nodes of a graph without an upper bound on the number of colors. $\DC$ is the (potentially non proper) $degree+1$ coloring problem, i.e., the color $c(u)$ of node $u$ has to be in the range $\set{1,\dots,\deg(v)+1}$. Both problems are LCL problems as the feasibility of a solution can be checked by investigating the $1$-neighborhood. This section is devoted to proving \Cref{cor:mainColoring}.

\mainColoring*

For this purpose we will present two randomized algorithms, one being $\T$-dynamic (cf. \Cref{sec:dcolor}) and the other $(\T,\alpha)$-network-static for $(\IC,\DC)$, w.h.p., for a $\T\in \bigO(\log n)$ and $\alpha=2$ (cf. \Cref{sec:scolor}). Both algorithms are variants of the following basic randomized coloring algorithm \cite{barenboim12,johansson99} that operates in phases of two rounds: In the first round each uncolored node $v$ chooses a tentative uniformly at random color from the range $\set{1,\dots,\deg(v)+1}\setminus S$ (where $S$ is the set of forbidden colors that colored neighbors have chosen previously). In the second round $v$ keeps the color if no neighboring node picked the same color and otherwise it discards the color. This two rounds in one phase implementation does not allow asynchronous wake-ups. Instead we provide a pipelined version in which all rounds are identical and a common global round counter is not needed. Thus our algorithm works in the asynchronous wake-up model.

\subsection{\texorpdfstring{\boldmath The $O(\log n)$-Dynamic Coloring Algorithm \DColor}{The O(log n)-Dynamic Coloring Algorithm \DColor}}
\label{sec:dcolor}
\textbf{\boldmath \DColor} is a variant of the basic randomized coloring algorithm, with the difference that the communication network is always restricted to the current intersection graph. At all times each uncolored node has a palette $P_v$ of potential colors. When \DColor is started in round $j$, the palette $P_v$ is initialized with  the set $[d_j(v)+1]$ without the colors of $v$'s neighbors in $G_j$.  As long as $v$ is uncolored, in each round $r\geq j$ it chooses a tentative uniformly at random color from its current palette, sends it to its neighbors and receives the tentative colors and permanently chosen colors from its neighbors in the intersection graph $G_{j+r}^{r\cap}$. If its tentative color $c$ is not among the received colors, $v$ permanently keeps color $c$ and informs its neighbors about its choice in the next round. Otherwise, $v$ stays uncolored, deletes the received permanent colors from its palette and repeats the procedure.

\begin{algorithm}[H]
\caption{\DColor}
\small
\begin{tabular}{@{}lll}
\textbf{Input:} & $n$-vector $\phi$ & \\
\textbf{Output:} & $\phi$ & \\
\textbf{Vars.: } & $\phi, P_v$ & Color palette.\\
& $\Tcount$ & Communication is always restricted to $G_r^{\Tcount\cap}$.
\end{tabular}
\medskip

\textit{Denote by $j$ the round in which the algorithm starts. $v$ does not have to know this value.}\\
\begin{tabular}{@{}lll}
\textbf{Start:} & $\Tcount=0$. \NoAlignComment{The start needs one communication round}\\
& Send $\phi_v$ to all neighbors in $G_j$. Receive values from neighbors.\\
& If $\phi_v=\bot$, set $P_v=[d_j(v)+1]\setminus \{\phi_w\mid w\in\NH_{G_j}(v)\}$  \NoAlignComment{Initialize color palette}
\end{tabular}
\algheading{Round $r$ of \DColor}
\begin{algorithmic}[1]
\Switch{$\phi_v=~$?}
\Case{$\phi_v=\bot$}
Pick tentative color $c_v\in P_v$ u.a.r. and send it to neighbors in $G_r^{\Tcount\cap}$.
\EndCase
\Case{$\phi_v\neq\bot$}
Send $\phi_v$ to neighbors in $G_r^{\Tcount\cap}$.
\EndCase
\EndSwitch
\State Receive fixed colors $F_v=\{\phi_w\mid w\in\NH_{G_r^{\Tcount\cap}}(v)\}$ and tentative colors $S_v=\{c_w \mid w\in\NH_{G_r^{\Tcount\cap}}(v)\}$.
\State $P_v=P_v\setminus F_v$ \NoAlignComment{Update color palette}
\Switch{$\phi_v=~$?}
\Case{$\phi_v=\bot$}
\IfThenElse{$c_v\in P_v$ and $c_v\notin S_v$}{$\phi_v=c_v$}{keep $\phi_v=\bot$.}
\EndCase
\Case{$\phi_v\neq\bot$}
Do nothing.
\EndCase
\EndSwitch
\State $\Tcount++$ \NoAlignComment{Intersect one more graph in the next round}
\State\textbf{Output }{$\phi$}
\end{algorithmic}
\label{alg:DColor}
\end{algorithm}

We show that \DColor is $O(\log n)$-dynamic, w.h.p.

\begin{lemma}
\label[lemma]{lem:A1A2forColoring}
$\DColor$ is $T$-dynamic for $(\IC,\DC)$ w.h.p. for a $\T\in \bigO(\log n)$.
\end{lemma}
We need to show that $\DColor$ has properties $A.1$ and $A.2$ (cf. \Cref{def:dynamic/static}). Property $A.1$ follows immediately, for property $A.2$ we show that despite the dynamics, w.h.p., all nodes are colored after $\bigO(\log n)$ rounds (\Cref{lem:runtimeColoring}). For this purpose we prove that the palette of $v$ is always larger than the number of uncolored neighbors in the intersection graph (\Cref{lem:palettesize}). With this property we show in \Cref{lem:oneroundtrivialdynamic} that if less than a fourth of the colors are deleted from $v$'s palette in the current round, then with constant probability $v$ chooses a so called \emph{good} color that it can keep with constant probability. \Cref{lem:runtimeColoring} then follows together with the property that colors are never added to $v$'s palette in $\DColor$ and a node is colored once its palette size equals one.

If \DColor is started in round $j$, then for a node $v$ and an $\Tcount\geq 0$, let \[\UN(v):=\{u\in\NH_{G_{j+\Tcount}^{\Tcount\cap}}(v)\mid\phi_u=\bot\}\] be the set of uncolored neighbors of node $v$ in the intersection graph in round $j+\Tcount$. We omit the round number in the notation as it will be always clear from the context.

\begin{lemma} \label[lemma]{lem:palettesize}
For all $v\in V$, in every round of \DColor one has $|P_v|\geq |\UN(v)|+1$.
\end{lemma}

\begin{proof}
Assume \DColor is started in round $j$. The inequality is true in round $j$ as $P_v$ is initially set to $[d_j(v)+1]$. In the following rounds, whenever a color is removed from $P_v$, at least one neighbor of $v$ chose this color, i.e., $|\UN(v)|$ decreases by at least one. Apart from that, changes in the graph topology can only decrease the number of uncolored neighbors of $v$ in the intersection graph and do not affect the palette.
\end{proof}

\begin{lemma} \label[lemma]{lem:oneroundtrivialdynamic}
In one round of \DColor, each uncolored node is colored with probability at least $1/64$ or its color palette shrinks by a factor of at least $1/4$.
\end{lemma}

\begin{proof}
Assume \DColor is started in round $j$. Let $\Tcount\geq 0$ and $v$ be uncolored at the beginning of round $j+\Tcount$. We assume $|\UN(v)|\geq 1$ (otherwise, $v$ will be colored in the current round as there will be no conflicts for $v$'s color choice). As $v\in\UN(u)$ for all $u\in\UN(v)$ one deduces that $\UN(u)\geq 1$ holds for these nodes. By \Cref{lem:palettesize}, $v$ and its uncolored neighbors have palettes of size at least $2$. We emphasize that all of the following definitions and arguments are only for the sake of the analysis and nodes executing the algorithm do not need to know these parameters. 
For a color $c\in P_v$, define the weight of $c$ as \[w_c:=\sum\limits_{\{u\in\UN(v)\mid c\in P_u\}}\frac{1}{|P_u|}~.\]
Let $Z_v$ be the set of those colors in $P_v$ which have been permanently chosen by some $u\in\NH_{G_{j+\Tcount}^{\Tcount\cap}}(v)$ in the last round (these are the colors which will be deleted from $P_v$ in the current round after (!) node $v$ chose its tentative color). Call a color $c\in P_v$ \emph{good} if $c\notin Z_v$ and $w_c\leq 2$. For a good color $c$ we have
\[\Pr\left(v\text{ keeps }c\mid v\text{ chose $c$ as tentative color}\right)=\prod\limits_{\{u\in\UN(v)\mid c\in P_u\}}\left(1-\frac{1}{|P_u|}\right)\stackrel{(\ref{eqn:x4})}{\geq}4^{-w_c}\geq 4^{-2}=\frac{1}{16}~.\]

At most $|\UN(v)|/2$ colors from $P_v$ can have a weight larger than $2$ because
\[\sum\limits_{c\in P_v}w_c=\sum\limits_{u\in\UN(v)}\left(\sum\limits_{c\in P_u\cap P_v}\frac{1}{|P_u|}\right)=\sum\limits_{u\in N(v)}\frac{|P_u\cap P_v|}{|P_u|}\leq |\UN(v)|~.\]
So in addition to the colors in $Z_v$, at most $\frac{|\UN(v)|}{2}$ colors in $P_v$ are not good. With $\frac{|\UN(v)|}{2}\leq \frac{|P_v|}{2}$ it follows that in $P_v$, at least $|P_v|-|Z_v|-\frac{|P_v|}{2}$ colors are good.
When we assume that $|Z_v|\leq \frac{|P_v|}{4}$ (i.e., the color palette of $v$ shrinks by a factor of at most $1/4$ in this round), then at least one fourth of the colors, i.e., $|P_v|-|Z_v|-\frac{|P_v|}{2}\geq\frac{|P_v|}{4}$, are good. So in this case, the probability for choosing a good color is at least $1/4$, which means that the overall probability for $v$ being colored is at least $1/64$. Therefore, if the color palette of $v$ does not shrink by a factor of at least $1/4$, $v$ is colored with probability at least $1/64$.
\end{proof}

\begin{lemma}\label[lemma]{lem:runtimeColoring}
There is a $T\in \bigO(\log n)$ such that for any dynamic graph and any input, after $T-1$ rounds of $\DColor$, w.h.p., all nodes are colored.
\end{lemma}

\begin{proof}
Assume \DColor is started in round $j$ on some dynamic graph and some input $\phi$. Fix a constant $b\geq 1$ and set $T_1:=\log_{\frac{3}{4}}\left(\frac{4}{n}\right)$, $T_2:=64(b+1)\ln(n)$ and $T':=T_1+T_2=\bigO(\log n)$. For each initially uncolored node $v$ (i.e., $\phi_v=\bot$), denote by $A_v$ the event that $v$ is not colored after round $j+T'$, i.e., after the execution of the start round $j$ and the $T'$ following rounds. For $A_v$ to come true, there can have been at most $T_1$ rounds in which $v$'s color palette shrinks by a factor of at least $1/4$, because after $T_1$ such rounds, one has $|P_v|\leq n\left(\frac{3}{4}\right)^{T_1}=1$ (initially it is $|P_v|\leq n$), which means that $v$'s color palette can not shrink another time without $v$ getting colored (a node will get colored before its palette gets empty). By \Cref{lem:oneroundtrivialdynamic} it follows that there must have been at least $T_2$ rounds in which $v$ got colored with probability at least $1/64$, so we obtain $\Pr(A_v)\leq(1-\frac{1}{64})^{T_2}\leq e^{-\frac{T_2}{64}}=\frac{1}{n^{b+1}}$.
With a union bound over all nodes, we can upper bound the probability that there is an uncolored node left after round $j+T'$:
\[\Pr\left(\bigcup\limits_{u\in V}A_u\right)\leq\frac{n}{n^{b+1}}=\frac{1}{n^b}~.\]
It follows that with probability at least $1-\frac{1}{n^b}$, all nodes are colored after round $j+T'$. With $T:=T'+2$, we get the desired result.
\end{proof}

Once all nodes are colored the output of $\DColor$ will be the same in all following rounds as $\DColor$ never uncolors a node. 

Before we prove \Cref{lem:A1A2forColoring}, we shortly characterize a vector $\phi$ that is partial packing and partial covering in this context: A vector $\phi$ is partial packing if there is an extension of $\phi$ in which the LCL condition of $\IC$ is satisfied for all nodes with $\phi_v\neq \bot$ (cf. \Cref{def:partialsolution}). If the graph induced by all colored nodes of $\phi$ forms a proper coloring it is straightforward to build such an extension by greedily coloring the remaining nodes.  Thus $\phi$ is partial packing if and only if  the graph induced by all colored nodes of $\phi$ forms a proper coloring. A vector $\phi$ is partial covering if the LCL condition of $\DC$ is satisfied for all nodes with $\phi_v\neq \bot$ and for all extensions of $\phi$. The feasibility of the LCL condition of $\DC$ for a node $v$ only depends on the color of $v$ and its degree -- it is independent from the colors of its neighbors. Both parameters do not depend on the choice of the extension and it is sufficient that $v$'s color is in $[d(v)+1]$ for all $v$ with $\phi_v\neq \bot$ to prove that $\phi$ is partial covering.

\begin{proof}[Proof of \Cref{lem:A1A2forColoring}]
Let $T=\bigO(\log n)$ be as in \Cref{lem:runtimeColoring}.

\noindent\textbf{\boldmath Property $A.1$:} Getting $\phi$ as input, \DColor will only change the values of nodes $v$ with $\phi_v=\bot$. Hence \DColor is input-extending.

\noindent\textbf{\boldmath Property $A.2$:} Let $j\geq \T-1$ and $\phi$ a partial solution for $(\IC,\DC)$ in $G_{j-T+1}$.
Let $\phi':=\DColor_{j-\T+2}^j(\phi)$. By \Cref{lem:runtimeColoring}, all nodes have chosen a color, w.h.p. We show: $(1)$ $\phi'_v\neq\phi'_w$ for all nodes $v$ and $w$ adjacent in $G_j^{\cap \T}$; $(2)$ $\phi'_v\in [d_j^{\cup T}(v)+1]$ for all nodes $v$.

$(1)$: Consider two nodes $v$ and $w$ adjacent in every graph $G_{j-\T+1},\dots,G_j$, i.e., adjacent in $G_j^{\cap \T}$. If $\phi_v,\phi_w\neq\bot$, we also have $\phi'_v\neq \phi'_w$ as $\phi$ is partial packing for $\IC$ in  $G_{j-T+1}$ and both nodes keep their color. Now, assume that $v$ or $w$ is uncolored in round $j-T+1$. It is not possible for them to take the same color in the same round (if they choose the same tentative color, they discard it again). If node $v$ (node $w$) is colored with color $c$ before $w$ (before $v$) or was already  colored with color $c$ in the input $\phi$ , then $c$ is removed from $w$'s palette ($v's$ palette), i.e., node $w$ (node $v$) is not able to take  color $c$ in the following rounds. 

$(2)$: Fix a node $v$. If $\phi_v\neq\bot$, then $\phi'_v=\phi_v$ (A.1) and $\phi_v\in [d_{j-\T+1}(v)+1]\subseteq [d_j^{\cup T}(v)+1]$ as $\phi$ is partial covering for $\DC$ in $G_{j-T+1}$. If $v$ gets colored in some round $i\in \{j-\T+2,\dots,j\}$, it takes a color among $[d_i(v)+1]\subseteq [d_j^{\cup T}(v)+1]$. 
\end{proof}

\subsection{\texorpdfstring{\boldmath The $O(\log n)$-Network-Static Coloring Algorithm \SColor}{The O(log n)-Network-Static Coloring Algorithm \SColor}}
\label{sec:scolor}
\textbf{$\SColor$} is similar to \DColor and we describe a single round of the algorithm: Colored nodes send their color to their neighbors (call these colors \emph{fixed} to distinguish them from tentative colors), uncolored nodes choose a tentative color from their palette and send them to their neighbors. But here, unlike in \DColor, the graph used for communication in round $r$ is the current graph $G_r$ and not the intersection graph. Then the color palettes are updated: Node $v$'s new palette is the set $[d_r(v)+1]$  without  the fixed colors of its neighbors. So in contrast to \DColor, colors can also be added to the palette. Then two cases are considered:
(1) If $v$ is uncolored, it checks if its tentative color is part of its new palette and not among its neighbors' tentative colors. If \emph{yes}, $v$ colors itself with this color, if \emph{not}, $v$ stays uncolored. (2) If $v$ is colored, it checks if its color is part of its palette. If \emph{yes}, it keeps its color, if \emph{not}, it uncolors itself.

\begin{algorithm}[H]
\caption{\SColor}
\small
\begin{tabular}{@{}lll}
\textbf{Input:} & $\phi$ & \\
\textbf{Output:} & $\phi$ & \\
\textbf{Vars.: } & $\phi, P_v$ & Color palette.
\end{tabular}

\medskip

\textbf{Start:} $P_v=\{1\}$ \NoAlignComment{Initialize color palette (no communication round needed)}
\algheading{Round $r$ of \DColor}
\begin{algorithmic}[1]
\Switch{$\phi_v=~$?}
\Case{$\phi_v=\bot$}
Pick tentative color $c_v\in P_v$ u.a.r. and send it to neighbors in $G_r$.
\EndCase
\Case{$\phi_v\neq\bot$}
Send $\phi_v$ to neighbors in $G_r$.
\EndCase
\EndSwitch
\State Receive fixed colors $F_v=\{\phi_w \mid w\in\NH_{G_r}(v)\}$ and tentative colors $S_v=\{c_w \mid w\in\NH_{G_r}(v)\}$.
\State $P_v=[d_r(v)+1]\setminus F_v$ \NoAlignComment{Update color palette} \label{line:update}
\Switch{$\phi_v=~$?}
\Case{$\phi_v=\bot$}
\IfThenElse{$c_v\in P_v$ and $c_v\notin S_v$}{$\phi_v=c_v$}{keep $\phi_v=\bot$.}
\EndCase
\Case{$\phi_v\neq\bot$} \NoAlignComment{Potential Uncoloring}
\If{$\phi_v\notin P_v$} $\phi_v=\bot$. \label{line:uncolor}
\EndIf
\EndCase
\EndSwitch
\State\textbf{Output }{$\phi$}
\end{algorithmic}
\label{alg:StColor}
\end{algorithm}

We show that \SColor is $(\T,\alpha=2)$-network-static, w.h.p.
The result is based on the local nature of the classic proof and the fact that a node and its neighbors do not uncolor themselves in \SColor if the $2$-neighborhood of the node is static.

\begin{lemma} \label[lemma]{lem:colornetworkstatic}
\SColor is $(\T,\alpha=2)$-network-static for $(\IC,\DC)$ w.h.p. for a $\T\in \bigO(\log n)$.
\end{lemma}

\begin{proof}
For proving $B.1$ we have to show that at the end of each round $r$, no colored node has a neighbor in $G_r$ with the same color and the color of node $v$ is in the range $[d_r(v)+1]$. Both properties are fulfilled as any node which does not satisfy them is uncolored (cf. line \ref{line:uncolor} in \Cref{alg:StColor}). Property $B.1$ is satisfied independently of the choice of $\T$.

For proving $B.2$, let $\T'=\bigO(\log n)$ be the runtime of the basic coloring algorithm for static graphs (\Cref{lem:oneroundtrivial}) and set $\T:=\T'+2$. Let $v\in V_r$ and $r_2\geq r+T$ such that $G_l[N_2(v)]=G_{l'}[N_2(v)]$ for all $l,l'\in [r,r_2]$ (for $r_2<r+T$, $B.2$ holds trivially). A node may uncolor itself only if it becomes adjacent to a node which has the same color or if the value of its current degree plus one falls below its chosen color. Both things do not happen if $v$'s $1$-neighborhood is static. So if $v$ is colored after round $r+T$, it will keep its color at least until round $r_2$.

We show why $v$ is colored in $O(\log n)$ rounds, w.h.p., if its $2$-neighborhood is static from round $r$ on: In that case $v$ executes the same steps in \SColor as it does in the basic algorithm for static graphs (cf. \Cref{app:trivial}). This does not hold for all nodes in the dynamic network. However, the knowledge about the behavior of all nodes that we need to mimic the proof of \Cref{lem:oneroundtrivial}  can be reduced to three properties that have to hold as long as $v$ is uncolored. Thus we only have to prove the following three properties:
\begin{itemize}
\item[(1)] In all rounds in the interval $[r+2,r_2]$, the color palettes of $v$'s uncolored neighbors have size $\geq 2$.
\item[(2)] From round $r+1$ to $r_2$, the size of $v$'s color palette is at least the number of $v$'s uncolored neighbors.
\item[(3)] From round $r+1$ to $r_2$, colors may only be deleted from but never join $P_v$.
\end{itemize}
With (1) and (2), we show that if less than $|P_v|/4$ colors are deleted from $P_v$ in the current round, then with  probability $1/16$ $v$ chooses a so called \emph{good} color and keeps it with probability at least $1/4$. Thus the palette of $v$ shrinks by a constant factor or $v$ is colored with constant probability. With (3) it follows that $v$ has chosen a color after $O(\log n)$ rounds, w.h.p.

Even though the three statements seem to be trivially satisfied this needs careful arguments, e.g., (1) might not be satisfied in rounds $r$ and $r+1$.

\noindent\textbf{(1)} Let $w\in\NH(v)$ and $v,w$ be uncolored in some round in $[r+2,r_2]$. As the nodes do not uncolor themselves in all rounds in $[r+1,r_2]$, both nodes are already uncolored at the end of round $r$ (one cannot deduce that they are uncolored in the first competition for colors at the beginning of round $r$ as they could be colored in round $r-1$, at the beginning of $r$ and only become uncolored at the end of round $r$ due to a new  edge in round $r$). Then in round $r+1$, $P_w$ is updated to $[d_r(w)+1]\setminus F_w$ (the degree of $w$ does not change), where $F_w$ contains the colors  of $w$'s colored neighbors. As $v$ is an uncolored neighbor of $w$, it follows that $|F_w|\leq d_r(w)-1$ and therefore $|P_w|\geq d_r(w)+1-|F_w|\geq 2$. In the following rounds, a color can only be deleted from $P_w$ if one of $w$'s (already existing) neighbors takes this color. But as long as $w$ has $v$ as uncolored neighbor, the value of $d_r(w)+1$ is always at least larger by two than the number of its colored neighbors.

\noindent\textbf{(2)} In round $r$, $P_v$ is set to $[d_r(v)+1]\setminus F_v$. The size of $F_v$ is at most the number of $v$'s colored neighbors, so the size of $P_v$ is at least the number of $v$'s uncolored neighbors. In the following rounds, as $v$'s degree is static, a color may only be deleted from $P_v$ if at least one of its neighbors chooses this color, i.e., if $v$ looses at least one uncolored neighbor. On the other hand, the number of $v$'s uncolored neighbors can not increase as none of $v$'s neighbors uncolors itself as long as $v$'s $2$-neighborhood is static.

\noindent\textbf{(3)} As $v$'s degree remains static, a color may only join $P_v$ if a neighbor of $v$ uncolors itself and its color thus becomes available for $v$ again. But as pointed out above no neighbor of $v$ uncolors itself if $v$'s $2$-neighborhood is static.
\end{proof}

\subsection{Proof of \Cref{cor:mainColoring}}
\Cref{thm:mainPackingCovering} with the $\bigO(\log n)$-network-static \SColor for $(\IC, \DC)$ (cf. \Cref{lem:colornetworkstatic}) and the $\bigO(\log n)$-dynamic algorithm \DColor (cf. \Cref{lem:A1A2forColoring}) for $(\IC, \DC)$ implies the result.

\begin{remark}
The analysis of \DColor and \SColor does not require the adversary to have any obliviousness. Thus, all results in this section are valid even for an adaptive offline adversary, which knows the choice of random bits in any round in advance.
\end{remark}

\clearpage

\section{MIS in Highly Dynamic Networks}
\label{sec:algorithm}
Let $\IS$ be the independent set problem (packing) and $\DS$ be the dominating set problem (covering). Both problems are LCL problems as the feasibility of a solution can be checked by investigating  the $1$-neighborhood. The main objective of this section is to prove \Cref{cor:mainColoring}.

\mainMIS*

Instead of  the vector-notation from \Cref{sec:dynModel}, we use the more intuitive notion with dynamic set variables: Algorithms produce a tuple of sets $(M,D)$ with $M$ denoting the MIS-nodes and $D$ the dominated nodes. This notation can be easily translated into the vector-notation from \Cref{sec:dynModel} by setting the value of a node to $1$ if it is in $M$, to $0$ if it is in $D$ and to $\bot$ if it is in $V\setminus (M \cup D)$.

The algorithm in section \Cref{sec:dmis} is a modification of Luby's algorithm \cite{alon86,luby86abbrv}. Luby's algorithm proceeds in phases of two rounds: First each undecided node draws a random number and sends it to its neighbors. In the second round, nodes with the smallest number in their neighborhood join $M$ and inform their neighbors which then join $D$. We present a pipelined version of Luby's algorithm in which each round is identical such that it works in the asynchronous wake up model.

The network static algorithm in \Cref{sec:smis} is based on a modified and pipelined version of the MIS algorithm in \cite{ghaffari16}.
\subsection{\texorpdfstring{\boldmath The $O(\log n)$-Dynamic MIS Algorithm \Luby}{The O(log n)-Dynamic MIS Algorithm \Luby}}
\label{sec:dmis}
In \textbf{\Luby} (dynamic MIS), any form of communication in round $r\geq i$ (if \Luby is started in round $i$) ignores edges added by the adversary after round $i$, i.e., communication is restricted  to the graph $G_r^{\cap(r-i+1)}=G_i\cap G_{i+1}\cap \ldots \cap G_r$. More detailed: At all times each node is in exactly one of three sets, i.e., in the set $M$ of MIS-nodes, in the set $D$ of dominated nodes or in the set $U$, i.e., the node is undecided.
The algorithm can be started in round $i$ with any configuration of states such that $M$ forms an independent set of $G_i$ and each node in $D$ has a neighbor in $M$ in $G_i$. To identify the current communication graph $G_r^{\cap\Tcount}$, we introduce a parameter $\Tcount$ that is initialized with zero and raised in every round. \footnote{Note that a global parameter as $\Tcount$ is not needed as every node can simply keep track of the current set of edges/neighbors it still has to consider for communication.}

\mypara{Sending.} Each node $v\in M$ sends a \notification to all nodes that were neighbors in the last $\Tcount$ rounds, that is, to all neighbors of $v$ in the graph $G_r^{\Tcount\cap}$; each node $v\in U$ draws a random number and sends it to its neighbors that were neighbors in the last $\Tcount$ rounds.

\mypara{After Receiving.} Nodes that receive a \notification change their state to \dominated. Still \undecided nodes that drew a smaller number than all random numbers they received join $M$. At the end of the round, the parameter $\Tcount$ is increased by one. 

\mypara{Output.} The algorithm returns the state of each node at the end of each round, i.e., it either returns \mis, \dominated or \undecided.

\begin{algorithm}[H]
\caption{\Luby}
\small
\begin{tabular}{@{}lll}
\textbf{Input:} & $(M,D)$  & (independent set, dominated nodes)\\
\textbf{Output:} & $(M, D)$ & \\
\textbf{Vars.: } & $M,D,U$ &  MIS-nodes, dominated nodes, undecided nodes\\
    		     & $\Tcount$ & Communication is always restricted to $G_r^{\Tcount\cap}$ \\
\end{tabular}

\medskip

\textbf{Start:} $\Tcount=0$ \NoAlignComment{No communication round needed}

\algheading{Round $r$ of \Luby}
\begin{algorithmic}[1]
\State $U=V\setminus (M\cup D)$
\Switch{$v\in~?$}
\Case{$v\in M$} 
send \notification to all neighbors in $G_r^{\T\cap }$  
\EndCase
\Case{$v\in U$}
send random number $\alpha_v\in[0,1]$ to all neighbors in $G^{\Tcount\cap}_r$ 
\EndCase
\EndSwitch
\State Receive random numbers $\{\alpha_u \mid u\in U\cap \NH_{G^{\Tcount\cap}_r}(v)\}$ and marks from all neighbors in $G^{\Tcount\cap}_r$
\Switch{$v\in~?$}
\Case{\notification received}
join $D$  
\EndCase
\Case{$v\in U$ and $\alpha_v<\min\{\alpha_u \mid u\in U\cap \NH_{G^{\Tcount\cap}_r}(v)\}$} 
join $M$ 
\EndCase
\EndSwitch
\State $\Tcount++$ \NoAlignComment{Intersect one more graph in the next round}
\State\textbf{Output }{$(M,D)$}
\end{algorithmic}
\label{luby}
\end{algorithm}

\begin{lemma} \label[lemma]{thm:A1A2forMIS}
$\Luby$ is $\T$-dynamic for $(\IS,\DS)$, w.h.p., for a $\T\in \bigO(\log n)$.
\end{lemma}

First we prove that there is a $\T\in \bigO(\log n)$ such that after $\T-1$ rounds of $\Luby$, w.h.p., all nodes are decided, i.e., either joined $M$ or $D$. The proof is similar to the 'standard' Luby analysis in \cite{windsor04,Yves10}, but needs additional care due to the dynamicity of the graph. As the graph changes, edges which are needed to inform neighbors about a joining MIS node might not be there anymore in the next round and the proofs in \cite{windsor04,Yves10} heavily rely on these edges. We adapt the proof and show that within two rounds, in expectation one third of the edges between undecided nodes are removed in the intersection graph, either because the adversary removes the edge or because one (or both) endpoints join $M$ or $D$.

To ease presentation, we write $M_r$, $D_r$ and $U_r$ for the state of the set variables $M$, $D$, $U$ in $\Luby$ \underline{at the beginning} of round $r$. Furthermore we define $H_r:=G_{j+r}^{r\cap}[U_r]$ for each positive integer $r$.

\begin{lemma} \label[lemma]{onethird}
Given a dynamic graph, assume \Luby is started in round $j$ on some input $(M,D)$. Then for each $r\geq 0$ we have
\[\EV\left[|E(\Htwo)|~\big|~|E(\Hzero)|\right]\leq \frac{2}{3}~|E(\Hzero)|~.\]
\end{lemma}

\begin{proof} We will show that, in expectation, at least one third of the edges of $\Hzero$ are not contained in $\Htwo$, i.e, we show $\EV\left[|E(\Hzero)\setminus E(\Htwo)|~\big|~|E(\Hzero)|\right] \geq \frac{1}{3}|E(\Hzero)|$. This is sufficient to prove the claim as $E(\Htwo)\subseteq E(\Hzero)$. 
Therefore define the following set of edges
\[E':=\{\{v,w\}\in E(\Hzero)\mid \{v,w\}\in E_{r+1}\text{ and } v,w \notin D_{r+1}\}~.\]
For $\{v,w\} \in E'$  define the event $(v \rightarrow w)_r$ as
$\alpha_v<\alpha_x\text{ for all }x\in\NH_{\Hzero}(v) \cup \NH_{\Hzero}(w)\setminus \{v\})~.$
We consider two (non-disjoint) types of edges of $\Hzero$ that are not contained in $\Htwo$: 
\begin{enumerate}
\item Edges that are not contained in $E'$,
\item Edges that are removed due to an event $(v \rightarrow w)_r$ for some $\{v,w\} \in E'$.
\end{enumerate}

In the following we lower bound the expected number of edges of type (2) by $|E'|/2$. 
The event $(v \rightarrow w)_r$ says that the chosen random value of $v$ in round $r$ is smaller than those of all its neighbors in $\Hzero$ as well as those of all neighbors of $w$ in $\Hzero$ (without $v$ of course).
If the event $(v \rightarrow w)_r$ occurs, $v$ joins $M$ in round $r$ and so $w$ joins $D$ in round $r+1$. 
Note that $v$ actually does join $M$ as $v\notin D_{r+1}$ implies that $v$ did not receive a mark from a neighbor in $M$ in round $r$. Furthermore node $w$ does indeed join $D$ because $\{v,w\}\in E_{r+1}$ so that $v$ can actually inform $w$ about joining $M$.
Thus event $(v \rightarrow w)_r$ implies that all  incident edges of $w$ in $\Hzero$ (i.e., $d_{U_r}(w):=|\Gamma_{\Hzero}(w)|$ many) will not be contained in the graph $\Htwo$. 

Let $X_{(v \rightarrow w)_r}$ be the random variable with value $d_{U_r}(w)$ if event $(v \rightarrow w)_r$ occurs and $0$ otherwise. Then  $X=\sum_{\{v,w\}\in E'}X_{(v \rightarrow w)_r}$ denotes the number of removed edges of type (2) (with some double counting involved). 
We now lower bound $\EV\left[X~\big|~|E(\Hzero)|\right]$.
As we have assumed that for its changes at the beginning of round $r+1$, the adversary can not take into account the choice of the random values made in round $r$, the probability of the event $(v \rightarrow w)_r$ for $\{v,w\} \in E'$ is at least $1/(d_{U_r}(v) + d_{U_r}(w))$. Thus we can lower bound as follows.
\begin{align*}
\EV\left[X~\big|~|E(\Hzero)|\right] &= \sum\limits_{\{v,w\}\in E'} \mathbb{E}[X_{(v \rightarrow w)_r}\big|~|E(\Hzero)|] + \EV[X_{(w \rightarrow v)_r}\big|~|E(\Hzero)|] \\
& =\sum\limits_{\{v,w\}\in E'} P[(v \rightarrow w)_r]\cdot d_{U_r}(w) + P[(w \rightarrow v)_r]\cdot d_{U_r}(v) \\
&\geq \sum\limits_{\{v,w\}\in E'} \frac{d_{U_r}(w)}{d_{U_r}(v)+d_{U_r}(w)}+\frac{d_{U_r}(v)}{d_{U_r}(w)+d_{U_r}(v)} 
=\sum\limits_{\{v,w\}\in E'} 1 = |E'|~.
\end{align*}
To determine the expected edges that are removed we now take care of double counting: If for an edge $\{v,w\}$ of $\Hzero$ there are $x,y$ with $\{x,v\} \in E'$ and $\{y,w\} \in E'$ such that both $(x \rightarrow v)_r$ and $(y \rightarrow w)_r$ holds, $\{v,w\}$ is counted twice in $X$. 
Every edge in $E(\Hzero)\setminus E'$ is not contained in the graph $\Htwo$. However, we might have counted such an edge twice in $X$ as follows; There might be $v,w,x,y \in U_r$ with $\{x,v\},\{v,w\},\{w,y\}\in E(\Hzero)$ such that $(x \rightarrow v)_r$, $(y \rightarrow w)_r$ and $\{v,w\}\notin E'$. Thus $\EV[X]+\big|E\left(\Hzero\right)\setminus E'\big|$ counts each removed edge between undecided nodes up to three times (once by $E\left(\Hzero\right)\setminus E'$ and twice by $X$). Hence the number of edges between undecided nodes that are removed in expectation is lower bounded by 
\begin{align*}
\frac{1}{3}\left(\EV[X]+\big|E\left(\Hzero\right)\setminus E'\big|\right)\geq \frac{1}{3}\left(\big|E'\big|+\big|E\left(\Hzero\right)\setminus E'\big|\right)=\frac{1}{3}\big|E\left(\Hzero\right)\big|~. & \qedhere
\end{align*} 
\end{proof}

\begin{remark}The proof of \Cref{onethird} relies on a 2-oblivious adversary: If the adversary knew the random values of round $r$, it could, e.g., delete all edges between nodes for which $(v \rightarrow w)_r$ holds. Therefore, the probability of event $(v \rightarrow w)_r$ would be zero for all $\{v,w\} \in E'$. 
\end{remark}

To show that the algorithm terminates in $O(\log n)$ rounds we need the following lemma.
\begin{lemma} \label[lemma]{RV}
For every $c \in [0,1)$ there is a $T \in \bigO(\log K)$ such that for every series of random variables $K \geq X_0\geq X_1\geq\dots$ with $\EV[X_{i+1}|X_i]\leq c\cdot X_i$, w.h.p., we have $X_T< 1$.
\end{lemma}
\begin{proof}[Proof sketch:]
For $0\leq i\leq T=O(\log n)$ define the random variable $Y_{i+1}$ that is $1$ if $X_{i+1}\leq (1/2+c/2) X_i$ or $X_{i+1}<1$. Due to Markov's inequality the probability that $Y_i$ equals $1$ is constant. We have $X_T<1$ if at least $\log_2 K$ of the $Y_i$'s have value one and it is straightforward to show that this holds w.h.p.
\end{proof}

\Cref{onethird,RV} imply the $\bigO(\log n)$ 'runtime' of $\Luby$.

\begin{lemma} \label[lemma]{runtimeMluby}
There is a $\T\in \bigO(\log n)$ such that for any dynamic graph and any input, after $T-1$ rounds of $\Luby$, w.h.p., all nodes are decided.
\end{lemma}

\begin{proof}
Assume \Luby is started in round $j$ on some dynamic graph $G_0,G_1,\dots$ with input $(M,D)$. For each $i\geq 0$, let $X_i$ be the random variable whose value is the number of edges in $G_{j+2i}^{2i\cap}$ with both endpoints in $U_{j+2i}$. We have $n^2\geq X_0\geq X_1,\dots$ and by \Cref{onethird} $\EV[X_{i+1}|X_i]\leq \frac{2}{3}X_i$ holds. By \Cref{RV}, there is a $T'\in \bigO\big(\log (n^2)\big)=\bigO\big(\log n\big)$ such that for any such sequence of random variables, w.h.p., $X_{T'}<1$ holds, which means that there are no two undecided nodes adjacent in $G^{2T'\cap}_{j+2T'}$. Then in round $j+2T'+1$, all leftover undecided nodes either join $D$ or $M$. The statement in the lemma then holds for $T:=2T'+2=\bigO(\log n)$.
\end{proof}

Note that \Cref{runtimeMluby} does not need any requirements on the input. However, in the proof of \Cref{thm:A1A2forMIS}, it is necessary that the input is a partial solution.
\begin{proof}[Proof of \Cref{thm:A1A2forMIS}]
Let $T\in \bigO(\log n)$ be chosen as in \Cref{runtimeMluby}. That is, after $T-1$ rounds of $\Luby$, w.h.p. all nodes are decided, i.e., all nodes are either in $M$ or $D$ and do not change their status afterwards. We now show that $\Luby$ is $\T$-dynamic, w.h.p.

\noindent\textbf{\boldmath Property $A.1$:} Getting $(M,D)$ as input, $\Luby$ only adds nodes to these sets but will never delete a node from them. Thus $\Luby$ is input-extending, which proves $A.1$. 

\noindent\textbf{\boldmath Property $A.2$:} Let $j\geq \T-1$ and the input $(M,D)$ be a partial solution for $(\IS,\DS)$ in the graph $G_{j-\T+1}$, i.e., $M$ is an independent set in $G_{j-\T+1}$ and every node in $D$ has a neighbor in $M$ in $G_{j-\T+1}$. Let $(M',D'):=\Luby_{j-\T+2}^j(M,D)$. We show the following two properties that imply property $A.2$ because, w.h.p., no node is undecided in round $j-\T+1$.

\noindent\textit{$M'$ is an independent set in $G_j^{\cap \T}$, deterministically.}
Consider two nodes adjacent in $G_j^{\cap \T}$, i.e., adjacent in every graph $G_{j-\T+1},\dots,G_j$. As the input set $M$ is an independent set in $G_{j-\T+1}$, they can not be both in $M$ in the beginning. Furthermore, they can not both join $M$ in the same round as they are part of each others neighborhood and only a node with the smallest random value among its neighbors may join $M$. If one of the nodes joins $M$ in some round (or already was in $M$ in the beginning), then the other one joins $D$ latest in the next round. As nodes do not leave $M$ or $D$ again, there is no way for the two nodes to be both in $M$ in round $j$. Therefore, $M'$ is an independent set in $G_j^{\cap \T}$ and this holds deterministically as it is not required that all nodes are decided. 

\noindent\textit{$M'$ is a dominating set in $G_j^{\cup \T}$, w.h.p.} If a node is in $D$ at the beginning, it is adjacent to a node in $M$ in $G_{j-\T+1}$, as the input $(M,D)$ is a partial solution for $(\IS,\DS)$ in $G_{j-\T+1}$. An initially undecided node joins $D$ in one of the rounds $j-\T+2,\dots,j$ only if it has been adjacent to a node in $M$ in some graph $G_{j-\T+2},\dots,G_j$. As nodes do not leave $M$, it follows that any node which is in $D'$ has been adjacent to a node in $M'$ in some graph $G_{j-\T+1},\dots,G_j$, i.e., in the graph $G_j^{\cup T}$. In round $j$, w.h.p., there is no undecided node left, i.e., all nodes not in $M'$ are in $D'$. It follows that $M'$ is, w.h.p., a dominating set in $G_j^{\cup \T}$.
\end{proof}

\subsection{\texorpdfstring{\boldmath The $O(\log n)$-Network-Static MIS Algorithm \StableMIS}{The O(log n)-Network-Static MIS Algorithm \StableMIS}}

\label{sec:smis}

The framework presented in \Cref{sec:packingCovering} applied to the MIS problem starts a new instance  of $\Luby$ in every round. After $T$ rounds it outputs the oldest instance and discards it. These instances should not start to compute a solution from scratch if the dynamic graph does not change much. Instead each instance begins its computation with a backbone independent set that does locally not change if the graph is locally static. The algorithm that computes the backbone independent has to be network static (cf. \Cref{def:dynamic/static}) and here we present and analyze the network static algorithm $\StableMIS$ (static MIS).
It is strongly influenced by Ghaffari's algorithm \cite{ghaffari16} with the crucial difference that nodes can leave the set of MIS nodes and become undecided again. \StableMIS uses the current graph for all communication.

\medskip

\textbf{\StableMIS:} At all times, each node is in exactly one of three sets: In the set $M$ of MIS-nodes, in the set $D$ of dominated nodes or in the set $U$ of undecided nodes. Each node $v$ has a \emph{desire-level} $p(v)$ which is initially set to $1/2$ and is updated in every round. It is upper bounded by $1/2$ and \ch{lower bounded by $1/(5n)$}\footnote{In \cite{ghaffari16}, desire-levels do not need a lower bound. However, in the dynamic setting, we need to avoid that desire-levels can become arbitrary small.}. In each round, the \emph{effective-degree} of $v$, $\delta(v)$, is set to $\delta(v)=\sum_{u\in\NH(v)\cap U}p(u)$ and is used to update $v$'s desire-level. As long as $v$ is decided, its desire-level does not influence the algorithm and thus it is not updated until $v$ becomes undecided again. $p_r(v)$ ($\delta_r(v)$) denotes the desire-level (effective-degree) of $v$ at the beginning of round $r$ before they are updated in the course of the round.

\mypara{Sending.}At the start of round $r$ each node $v$ in $M$ sends a \notification to all neighbors in $G_r$; each node $v$ in $U$ becomes a \emph{candidate} with probability $p_r(v)$ and sends $p_r(v)$ and the information whether it became a candidate to its neighbors in $G_r$.

\mypara{After Receiving.}
Undecided nodes update their desire-level. Nodes that were in state \undecided and received a \notification join the set $D$. Still undecided nodes that became a candidate and have no neighbor that is also a candidate join $M$. Nodes in $M$ that received a \notification leave $M$ and become undecided. Nodes in $D$ which loose their domination, i.e., do not receive a \notification in the current round, become undecided.

\mypara{Output.} The algorithm returns the state of each node at the end of each round.

\begin{algorithm}[H] 
\caption{\StableMIS}
\small
\begin{tabular}{@{}lll}
\textbf{Input:} & $(M,D)$ & (independent set, dominated nodes)\\
\textbf{Output:} & $(M,D)$ & \\
\textbf{Vars.: } & $M,D,U$ &  MIS-nodes, dominated nodes, undecided nodes\\
                 & $p(v)$, $\delta(v)$ & desire-level, effective degree\\
\end{tabular}

\medskip

\textbf{Start:} $p(v)=1/2$ \NoAlignComment{No communication round needed}

\algheading{Round $r$ of \StableMIS}
\begin{algorithmic}[1]
\State $U=V\setminus (M\cup D)$
\Switch{$v\in ~?$}
\Case{$v\in M$} 
send \notification to all neighbors in $\NH_{G_r}(v)$ \label{receivemark}
\EndCase
\Case{$v\in U$}
Become \emph{candidate} with probability $p(v)$. \newline 
 \hphantom{Case $v\in U$:......} Send $p(v)$ and \emph{candidate-note} to all neighbors in $G_r$.
\EndCase
\EndSwitch
\State Receive marks, desire-levels and candidate-notes from all neighbors in $G_r$.

\noindent If $v\in U$: \NoAlignComment{update desire-level}

$\delta(v)=\sum\limits_{u\in U\cap\NH_{G_r}(v)}p(u)$

$p(v) = \begin{cases}
\max\{p(v)/2,\frac{1}{5n}\} & \text{if } \delta(v)\geq 2\\
\min\{2p(v),1/2\} & \text{if } \delta(v)<2
\end{cases}$
\label{line:updatedl}
\Switch{$v\in~ ?$}
\Case{$v\in U$ and \notification received}
join $D$
\EndCase
\Case{$v\in U$, no \notification received, being candidate and no candidate-note received} 
join $M$ \label{joinM}
\EndCase
\Case{$v \in M$  and  \notification received}
join $U$ \label{leaveMIS}
\EndCase
\Case{$v\in D$ and no \notification received}
join $U$ \label{leaveD}
\EndCase
\EndSwitch
\State\textbf{Output }{$(M,D)$}
\end{algorithmic}
\label[algorithm]{alg:StableMIS}
\end{algorithm}

\begin{lemma}\label[lemma]{lem:MISnetworkstatic}
\StableMIS is $(\T,\alpha=2)$-network-static for $(\IS,\DS)$, w.h.p., for a $\T\in\bigO(\log n)$.
\end{lemma}
We show that after $\bigO(\log n)$ rounds of \StableMIS, a node $v$ is decided, w.h.p., if its $2$-neighborhood is static, and does not change its output as long as its $2$-neighborhood is static. 
The core ideas of the proof are contained in the purely local analysis of \cite{ghaffari16}. However, the proof needs to be adapted in several places. Most important is the change of the definition of \emph{golden rounds of type two} that is needed because we use a pipelined version of the algorithm in \cite{ghaffari16}. Further, we need a more careful reasoning due to the facts that a node might not have desire level $1/2$ when its neighborhood becomes static and the cap of desire levels at $1/5n$.

\begin{lemma} \label[lemma]{thm:runtimeStableMIS}
After $\bigO(\log n)$ rounds of \StableMIS, a node $v$ is decided, w.h.p., if its $2$-neighborhood is static, and does not change its output as long as its $2$-neighborhood is static.
\end{lemma}

\begin{proof}
We first show that a node $v$ does not change its output once it is $\neq \bot$ and if its $2$-neighborhood is static.
A node $v$ leaves $M$ only if it becomes adjacent to a node that has  already been in $M$ as well. This does not happen if $v$'s $1$-neighborhood is static. Node $v$ may only leave $D$ if it either looses a neighbor (due to a graph change) which was in $M$ or if some neighbor $w$ of $v$ leaves $M$. The first will not happen if $v$'s $1$-neighborhood does not change, the second, as seen before, does not happen if $w$'s $1$-neighborhood is static, which is the case if $v$'s $2$-neighborhood is static.

To show that a node is decided fast, w.h.p., let $c\geq 1$. We show that there is a $\beta$ such that when \StableMIS is started in any round $j$ with some input on a dynamic graph which is static in the $2$-neighborhood of a node $v\in V_j$, the probability that $v$ is still undecided after $\beta\log n$ rounds is at most $\frac{1}{n^c}$.

We say an undecided node $u$ is \emph{low-degree} if $\delta_r(u)<2$, and \emph{high-degree} otherwise. We define two types of \emph{golden rounds} for an undecided node $v$: 
\begin{enumerate}
\item[(1)] rounds in which $\delta_r(v)<2$ and $p_r(v)=1/2$,
\item[(2)] rounds in which $\delta_r(v)\geq 1$ and at least $\delta_r(v)/10$ of it is contributed by low-degree neighbors \ch{that did not receive a mark at the beginning of the round}.
\end{enumerate}
For the sake of analysis, we assume that node $v$ keeps track of the number of golden rounds of each type it has been in. We show the following two statements:

\smallskip
\noindent\textbf{I.} After $\beta\log n$ rounds, either $v$ is decided or at least one of its two golden round counts reached $\frac{\beta}{13}\log n$.

\noindent\textbf{II.} If $v$ is still undecided and $r$ a golden round, with probability at least $1/200$, $v$ gets decided in round $r$ or $r+1$.

\medskip
If $r$ is a golden round, we call the rounds $r$ and $r+1$ a \emph{golden phase}. By \textbf{II}, the probability that $v$ does not get decided in a golden phase is at most $(1-1/200)$. Among $\frac{\beta}{13}\log n$ golden rounds, there are at least $\frac{\beta}{26}\log n$ golden phases (in the worst case, both rounds of a golden phase are golden rounds).

It follows that the probability that $v$ is not decided after at least $\frac{\beta}{13}\log n$ golden rounds (that is by \textbf{I} after $\beta\log n$ rounds) is at most\[\left(1-\frac{1}{200}\right)^{\frac{\beta}{26}\log n}\leq e^{-\frac{\beta}{200\cdot 26}\log n}\leq \left(\frac{1}{n}\right)^{\frac{\beta}{200\cdot 26}}\leq\frac{1}{n^c}\quad\text{for}\quad\beta\geq 200\cdot 26\cdot c~.\]

Recall that $U_r$ denotes the set of undecided nodes at the beginning of round $r$. We write $U_r(v)$ for the set of \emph{undecided neighbors} of $v$ at the beginning of round $r$. A decided node will not get undecided as long as its $1$-neighborhood is static, so a decided neighbor of $v$ does not get undecided if $v$'s $2$-neighborhood is static. Hence, one has $U_{r+1}(v)\subseteq U_r(v)$ (this is also valid for the algorithm in \cite{ghaffari16}, where decided nodes are considered as removed from the graph).

\medskip

\noindent\textbf{Proof of I.} We consider the first $\beta\log n$ rounds after the start of \StableMIS. Let $g_1$ and $g_2$ be the number of golden rounds of type 1 and type 2, respectively, during that period. We assume that after $\beta\log n$ rounds, $v$ is not decided and $g_1\leq \frac{\beta}{13}\log n$, and show that $g_2\geq\frac{\beta}{13}\log n$.

Let $h$ be the number of rounds in which $\delta_r(v)\geq 2$. We first show that if less than $1/13$ of the $\beta\log n$ rounds are $g_1$ rounds then almost half of the rounds (actually a $6/13$ fraction of the rounds minus an absolute value of $\frac{1}{2}\log\frac{5}{2}n$) are rounds with $\delta_r(v)\geq 2$, i.e., we lower bound $h$. Then, in a second step we show that $h$ is upper bounded by a function of $g_2$ which then implies the desired lower bound for $g_2$. 

We first lower bound  $h$. If $\delta_r(v)<2$, either $p_r(v)=1/2$, which means that $r$ is a type-1 golden round and $p_r(v)$ does not change, or $p_r(v)<1/2$ and $p_r(v)$ will increase by a factor of 2 (capped by $1/2$). As $p_r(v)$ is between $\frac{1}{5n}$ and $1/2$, the number of factor 2 increases of $p_r(v)$ is at most the number of factor 2 decreases plus $\log\frac{5}{2}n$ (every factor 2 decrease cancels one factor 2 increase, leaving at most $\log\frac{5}{2} n$ further increases without exceeding $1/2$). So if we take the total of $\beta\log n$ rounds, subtract the type-1 golden rounds and the '$\log\frac{5}{2} n$-slack', in at most half of the remaining rounds, the desire-level can increase, because for each increasing round there must be a decreasing counterpart. This means that there are at least $\frac{1}{2}(\beta\log n-g_1-\log\frac{5}{2}n)$ decreasing rounds, i.e., rounds with $\delta_r(v)\geq 2$. It follows 
\begin{align}\label{eqn:hlowerbound}
h\geq\frac{1}{2}(\beta\log n-g_1-\log\frac{5}{2}n)\stackrel{g_1\leq\frac{\beta}{13}}{\geq}\frac{6\beta}{13}\log n-\frac{1}{2}\log\frac{5}{2}n~.
\end{align}

Now upper bound $h$ with a function of $g_2$. If $\delta_r(v)\geq 1$ and $r$ is not a type-2 golden round, one has
\[\delta_{r+1}(v)=\sum_{u\in U_{r+1}(v)}p_{r+1}(u)=\sum_{\substack{u\in U_{r+1}(v) \\ \delta_r(u)<2}}p_{r+1}(u)+\sum_{\substack{u\in U_{r+1}(v) \\ \delta_r(u)\geq 2}}p_{r+1}(u)\]
For $u\in U_{r+1}(v)$ with $\delta_r(u)<2$ we have $p_{r+1}(u)\leq 2p_r(u)$ . For $u\in U_{r+1}(v)$ with $\delta_r(u)\geq 2$ one has either $p_{r+1}(u)=\frac{1}{2}p_r(u)$ or $p_{r+1}(u)=\frac{1}{5n}$. As $v$ has at most $n$ neighbors, the contribution of nodes $u$ with $p_{r+1}(u)=\frac{1}{5n}$ to $\delta_{r+1}(v)$ is at most $1/5$. Hence we get
\begin{align}
\label[equation]{eqn:effdegree}
\delta_{r+1}(v)\leq 2\sum_{\substack{u\in U_{r+1}(v) \\ \delta_r(u)<2}}p_r(u)+\frac{1}{2}\sum_{\substack{u\in U_{r+1}(v) \\ \delta_r(u)\geq 2}}p_r(u)+\frac{1}{5}~.
\end{align}
As $r$ is not a type-2 golden round and $\delta_r(v)\geq 1$, one has 
\begin{align} \label[equation]{eqn:simple}
& \sum\limits_{\substack{u\in U_{r+1}(v) \\ \delta_r(u)<2}}p_r(u)\leq\frac{1}{10}\delta_r(v), &   \sum\limits_{\substack{u\in U_{r+1}(v) \\ \delta_r(u)\geq 2}}p_r(u)\leq\delta_r(v), & & \frac{1}{5}\leq\frac{1}{5}\delta_r(v)~.
\end{align}
Using the inequalities in (\ref{eqn:simple}) to upper bound the terms in \Cref{eqn:effdegree} we obtain
\[\delta_{r+1}(v)\leq 2\frac{1}{10}\delta_r(v)+\frac{1}{2}\delta_r(v)+\frac{1}{5}\delta_r(v)=\frac{4}{5}\delta_r(v)~.\]
Thus every round in which $\delta_r(v)$ is increased to a value larger than $2$ is a  type-2 golden round. The effect of these rounds is canceled by at most four rounds with $\delta_r(v)\geq 2$ which are not golden rounds of type 2, because in these rounds, the effective degree shrinks by a factor of at least $4/5$ whereas in a golden round it is increased by at most a factor of $2$ and $(4/5)^4\cdot 2<1$. Apart from these $5g_2$ rounds that are either golden rounds or cancellation counterparts, in every remaining round with $\delta_r(v)\geq 2$, the effective degree is decreased at least by a $4/5$ factor. Due to $\delta_r(v)\leq n$ this can happen at most $\log_{\frac{5}{4}}n$ times while keeping $\delta_r(v)\geq 2$~.
It follows that the number of rounds with $\delta_r(v)\geq 2$ is at most $\log_{\frac{5}{4}}n+5g_2$, i.e., $h\leq\log_{\frac{5}{4}}n+5g_2$. Together with $h\geq\frac{6\beta}{13}\log n-\frac{1}{2}\log\frac{5}{2}n$ (\Cref{eqn:hlowerbound}) we get $g_2\geq \frac{\beta}{13}\log n$.

\medskip

\noindent\textbf{Proof of II.} In a type-1 golden round, $v$ becomes a candidate with probability $1/2$. The probability that no undecided neighbor of $v$ becomes a candidate is \[\prod\limits_{u\in U(v)}(1-p_r(u))\stackrel{(\ref{eqn:x4})}{\geq} 4^{-\sum_{u\in U(v)}p_r(u)}\geq 4^{-\delta_r(v)}\geq 4^{-2}=\frac{1}{16}~.\]So if $v$ does not receive a mark in round $r$, it joins $M$ with probability at least $1/32>1/200$ in a type-1 golden round. Otherwise (if $v$ receives a mark), it joins $D$ in the next round.

Let $L$ be the set of $v$'s undecided, low-degree neighbors that did not receive a mark. In a type-2 golden round, the probability that there is a candidate in $L$ is at least \[1-\prod\limits_{u\in L}(1-p_r(u))\geq 1-e^{-\sum_{u\in L}p_r(u)}\geq 1-e^{-\delta_r(v)/10}\geq 1-e^{-1/10}>0.08~.\]If $u\in L$ is a candidate, the probability that no undecided neighbor of $u$ is also a candidate is at least $\prod_{w\in U(u)}(1-p_r(w))\geq 4^{-\delta_r(u)}>\frac{1}{16}$. In this case, as $u$ did not receive a mark at the beginning of the round, it will join $M$ (line \ref{joinM}) and therefore $v$ joins $D$ in the next round, because $v$ and $u$ stay adjacent due to the assumption that the $1$-neighborhood of $v$ is static. So the probability that one of $v$'s neighbors joins $M$ (and thus $v$ joins $D$ in the next round) is at least $0.08/16=1/200$.
\end{proof}

Before we prove \Cref{lem:MISnetworkstatic}, we shortly describe what it means for an output vector $\phi$ to be partial packing and partial covering for the MIS problem. If there are no two adjacent nodes in the  state \emph{mis}, setting all nodes with $\phi_v=\bot$ to the  state \emph{dominated} yields an extension in which the LCL condition of $\IS$ is satisfied for all nodes with $\phi_v\neq\bot$. On the other hand, such an extension can only exist if no two \emph{mis} nodes are adjacent. Furthermore, the LCL condition of $\DS$ is satisfied if and only if every node in the state \emph{dominated} has an \emph{mis} neighbor. If, for some output $\phi$,  all nodes in the state \emph{dominated} already have an \emph{mis} neighbor their LCL condition is also satisfied for all extensions of $\phi$. Conversely, if one \emph{dominated} node in $\phi$ does not have an \emph{mis} neighbor in $\phi$ its LCL condition is violated in the extension which sets all nodes with $\phi_v=\bot$ to the  state \emph{dominated}. Thus, $\phi$ is partial packing for $\IS$ if and only if no two \emph{mis} nodes are adjacent and partial covering for $\DS$ if and only if every \emph{dominated} node is adjacent to an \emph{mis} node.

\begin{proof}[Proof of \Cref{lem:MISnetworkstatic}]
\noindent\textbf{\boldmath Property $B.1$:} If at the beginning of round $r$, there are two nodes in $M$ adjacent in $G_r$, both will receive a mark (cf. line \ref{receivemark} in \Cref{alg:StableMIS}) and leave $M$ (line \ref{leaveMIS}) by the end of the round. Additionally, a node from $U$ only joins $M$ if it did not receive a mark, i.e., if it has no neighbor in $M$ in $G_r$. Finally, no two adjacent nodes from $U$ join $M$ in the same round because a node only joins $M$ if it has no other candidate in its neighborhood. 
If at the start of round $r$, there is a node in $D$ not adjacent to a node in $M$, it will not receive a mark and leave $D$ (line \ref{leaveD}) by the end of the round. Additionally, a node from $U$ joins $D$ in round $r$ only if it received a mark, i.e., has a neighbor in $M$ in $G_r$. 
Property $B.1$ holds independently of the choice of $\T$.

\noindent\textbf{\boldmath Property $B.2$:} Let $v\in V_r$ and $\T=\bigO(\log n)$ be the time until $v$ is decided, w.h.p., if \StableMIS is started in round $r$ and $v$'s $2$-neighborhood is static (cf. \Cref{thm:runtimeStableMIS}). Let $r_2\geq r+T$ such that $G_l[N_2(v)]=G_{l'}[N_2(v)]$ for all $l,l'\in [r,r_2]$ (for $r_2<r+T$, $B.2$ holds trivially). Then, by \Cref{thm:runtimeStableMIS}, $v$ is decided after round $r+\T$, w.h.p., and does not change its output as long as its $2$-neighborhood is static, i.e., at least until round $r_2$.
\end{proof}

\subsection{Proof of \Cref{cor:mainMIS}}
\Cref{thm:mainPackingCovering} with the $\bigO(\log n)$-network-static algorithm \StableMIS for $(\IS, \DS)$ (cf. \Cref{lem:MISnetworkstatic}) and the $\bigO(\log n)$-dynamic algorithm \Luby (cf. \Cref{thm:A1A2forMIS}) for $(\IS, \DS)$ implies the result.
\begin{remark}The analysis of \Luby relies on a 2-oblivious adversary. 
\end{remark}

\section{Basic Randomized Coloring Algorithm for Static Graphs}
\label{app:trivial}

The following algorithm is a variant of the simplest randomized distributed node coloring algorithm (in static graphs). Usually (e.g. in \cite{barenboim12,johansson99}), this algorithm is described to operate in phases, each consisting of two communication rounds: In the first round of a phase, each uncolored node picks a tentative uniformly at random color $c$ from its palette and sends it to its neighbor. It keeps  $c$ as its permanent color if none of its uncolored neighbors chose $c$ in the same round. If it keeps a color permanently it informs its neighbors about its permanent color in the next round and deletes the received colors from its palette. Initially, the color palette of a node $v$ is set to $[d(v)+1]$.

With a slight adaption we can change this algorithm to an algorithm consisting of only one type of round. This has the advantage that this algorithm can also be used in an asynchronous wake-up model without needing a global clock that tells each node which type of round to execute.

\begin{algorithm}[H]
\caption{\Trivial}
\label{alg:basicColoring}
\small
\textbf{Start:} $\phi_v=\bot$, $P_v=\{1\}$ \NoAlignComment{No communication round needed}
\algheading{Round $r$ of \Trivial}
\begin{algorithmic}[1]
\Switch{$\phi_v=~$?}
\Case{$\phi_v=\bot$}
Pick tentative color $c_v\in P_v$ uniformly at random and send it to neighbors.
\EndCase
\Case{$\phi_v\neq\bot$}
Send $\phi_v$ to neighbors.
\EndCase
\EndSwitch
\State Receive fixed colors $F_v=\{\phi_w\mid w\in\NH(v)\}$ and tentative colors $S_v=\{c_w\mid w\in\NH(v)\}$.
\State $P_v=[d(v)+1]\setminus F_v$ \NoAlignComment{Update color palette}
\Switch{$\phi_v=~$?}
\Case{$\phi_v=\bot$}
\IfThenElse{$c_v\in P_v$ and $c_v\notin S_v$}{$\phi_v=c_v$}{keep $\phi_v=\bot$.}
\EndCase
\Case{$\phi_v\neq\bot$}
Do nothing.
\EndCase
\EndSwitch
\State\textbf{Output }{$\phi$}
\end{algorithmic}
\label{alg:trivialstatic}
\end{algorithm}

We believe that the following proofs are well known folklore but we could not find a publication that contains them.

\begin{lemma} \label[lemma]{lem:oneroundtrivial}
In one round of \Cref{alg:trivialstatic}, each node gets colored with probability at least $1/64$ or its color palette shrinks by a factor of at least $1/4$.
\end{lemma}

\begin{proof}
Fix some arbitrary round and an uncolored node $v$. Let $P_v$ be $v$'s color palette and $\UN(v)$ the set of $v$'s uncolored neighbors. We assume $|\UN(v)|\geq 1$ (otherwise, $v$ will be colored in the current round as there will be no conflicts with $v$'s color choice). It follows that also $\UN(u)\geq 1$ for all $u\in \UN(v)$, as $v\in \UN(u)$. For all nodes $w\in V$ one has $|P_w|\geq |\UN(w)|+1$. This holds at the beginning of the algorithm as $P_w$ is initially set to $[d(w)+1]$. In the following rounds, a color may only be removed from $P_w$ if it is taken by a neighbor. Therefore, $v$ and its uncolored neighbors have palettes of size at least $2$.

For $c\in P_v$, define the weight of $c$ as \[w_c:=\sum\limits_{\{u\in \UN(v)\mid c\in P_u\}}\frac{1}{|P_u|}~.\]
Let $Z_v$ be the set of those colors in $P_v$ which have been permanently chosen by a neighbor of $v$ in the last round (these are the colors which will be deleted from $P_v$ in the current round). Call a color $c\in P_v$ \emph{good} if $c\notin Z$ and $w_c\leq 2$. For a good color $c$ we have \[\Pr\left(v\text{ keeps }c\mid v\text{ chose $c$ as tentative color}\right)=\prod\limits_{\{u\in \UN(v)\mid c\in P_u\}}\left(1-\frac{1}{|P_u|}\right)\stackrel{(\ref{eqn:x4})}{\geq}4^{-w_c}\geq 4^{-2}=\frac{1}{4}.\]

Because of 
\[\sum\limits_{c\in P_v}w_c=\sum\limits_{u\in \UN(v)}\left(\sum\limits_{c\in P_u\cap P_v}\frac{1}{|P_u|}\right)=\sum\limits_{u\in \UN(v)}\frac{|P_u\cap P_v|}{|P_u|}\leq |\UN(v)|~,\]
it follows that at most $\frac{|\UN(v)|}{2}$ colors from $P_v$ can have a weight larger than $2$. So in addition to the colors in $Z_v$, at most $\frac{|\UN(v)|}{2}$ colors in $P_v$ are not good. With $\frac{|\UN(v)|}{2}\leq\frac{|P_v|}{2}$ it follows that in $P_v$, at least $|P_v|-|Z_v|-\frac{|P_v|}{2}$ colors are good.

When we assume that $|Z_v|\leq \frac{|P_v|}{4}$ (i.e., the color palette of $v$ shrinks by a factor of at most $1/4$), then at least $|P_v|-|Z_v|-\frac{|P_v|}{2}\geq\frac{|P_v|}{4}$ colors are good. So in this case, the probability for choosing a good color is at least $1/4$, which means that the overall probability for $v$ being colored is at least $1/64$. Therefore, if the color palette of $v$ does not shrink by a factor of at least $1/4$, $v$ gets colored with probability at least $1/64$.
\end{proof}

With \Cref{lem:oneroundtrivial} we can prove the $\bigO(\log n)$ runtime of \Cref{alg:trivialstatic}.

\begin{lemma} \label[lemma]{lem:runtimetrivial}
There is a $T\in \bigO(\log n)$ such that after $T$ rounds of \Cref{alg:trivialstatic}, w.h.p. all nodes are colored.
\end{lemma}

\begin{proof}
Fix a constant $b\geq 1$ and set $T_1:=\log_{\frac{3}{4}}\left(\frac{4}{n}\right)$, $T_2:=64(b+1)\ln(n)$ and $T:=T_1+T_2=\bigO(\log n)$. For each node $v$, denote by $A_v$ the event that $v$ is not colored after $T$ rounds. For $A_v$ to come true, there can have been at most $T_1$ rounds in which $v$'s color palette shrinks by a factor of at least $1/4$, because after $T_1$ such rounds, one has $|P_v|\leq n\left(\frac{3}{4}\right)^{T_1}=1$ (initially it is $|P_v|\leq n$), which means that $v$'s color palette can not shrink another time without $v$ getting colored (a node will get colored before its palette gets empty). By \Cref{lem:oneroundtrivial} it follows that there must have been at least $T_2$ rounds in which $v$ got colored with probability at least $1/64$, so we obtain $\Pr(A_v)\leq(1-\frac{1}{64})^{T_2}\leq e^{-\frac{T_2}{64}}=\frac{1}{n^{b+1}}$.
With a union bound over all nodes, we can upper bound the probability that there is an uncolored node left after $T$ rounds:
\[\Pr\left(\bigcup\limits_{u\in V}A_u\right)\leq\frac{n}{n^{b+1}}=\frac{1}{n^b}~.\]
It follows that with probability at least $1-\frac{1}{n^b}$, all nodes are colored after $T$ rounds.
\end{proof}

\section{Discussion}
\label{sec:discussion}
\subsection{A Simple Recipe for Developing Algorithms for Dynamic Graphs}
We believe that \Cref{sec:coloring} and \Cref{sec:algorithm} illustrate a general method to convert a distributed algorithm $\mathcal{A}$ with running time $T$ for a given static graph problem (which can be decomposed into a packing and covering problem) into a $T$-dynamic and a $(T,\alpha)$-network-static algorithm for the corresponding dynamic graph problem. For the $T$-dynamic algorithm, run $\mathcal{A}$ on the intersection graph (over all graphs since the start of the algorithm) and a node that generates an output keeps it in all following rounds. The correctness of such an algorithm is usually immediate; the analysis of the running time (the number of rounds until all nodes have an output $\neq\bot$) may need small adaptions that depend on the strength of the adversary.
For the $(T,\alpha)$-network-static algorithm, run $\mathcal{A}$ on $G_r$ as the communication graph in round $r$ with the additional property that at the end of the round, a node $v$ with output $\neq\bot$ gets undecided again if the partial packing or covering property is violated at $v$ (cf. \Cref{def:partialsolution} and assume that the LCL-radius of the problem is at most one such that a node can check whether to become undecided). This recipe seems promising to work for a wide range of local distributed graph algorithms. 

\subsection{Future Work}
For the MIS and coloring problem we found $T$-dynamic and $(T,2)$-network-static algorithms with window size $T=\bigO(\log n)$. This window size is optimal assuming that there are no faster algorithms for the static versions of the problems. However, for future research, one could allow more general use of this window. In the present algorithms the feasibility of an output $\phi$ depends on the topology of the last $T$ graphs in the dynamic graph sequence. In particular, output $\phi$ is feasible if it satisfies the packing constraint on the intersection graph $G^{\cap T}$ and the covering property on the union graph $G^{\cup T}$. Generalizing this feasibility definition to more general dependencies on the recent topology, e.g., only consider edges that have been there for a a $\delta$-fraction of the last $T$ rounds, with $\delta\in (0,1]$, is of interest.

In this paper, we assumed a round-based model, i.e., topological changes and sending messages are done in synchronous rounds. However, nodes do not need common knowledge of a round counter and, in particular, our algorithms work for asynchronous wake up. 
Algorithms with two or more types of rounds, e.g., the standard version of Luby's MIS algorithm alternates between competing rounds and notification rounds, do not immediately work with asynchronous wake up as nodes need to know the type of round when waking up (at least if it is necessary that nodes synchronously execute the same steps). Thus, to enable asynchronous wake up, we provided algorithms in which the nodes' behavior in every round is identical. An object of interest for future research would be considering an even higher extent of asynchronicity and removing the round-based model.

\clearpage
\bibliographystyle{alpha}
\bibliography{references}

\end{document}